\documentclass[12pt,one column]{IEEEtran}

\usepackage{setspace}
\pdfoutput=1
\usepackage{amsmath}
\usepackage{amssymb}
\usepackage{cite}
\usepackage{color}
\usepackage{epsfig}
\usepackage{epsf}
\usepackage{rotating}
\usepackage{mathrsfs}
\usepackage{epsfig}
\usepackage{graphics}
\usepackage{theorem}

\usepackage{amsmath}
\usepackage{amssymb}
\usepackage{cite}
\usepackage{color}
\usepackage{epsfig}
\usepackage{epsf}
\usepackage{rotating}
\usepackage{mathrsfs}
\usepackage{epsfig}
\usepackage{graphics}
\usepackage{bbold}
\usepackage{theorem}

\newtheorem{thm}{Theorem}
\newtheorem{lem}{Lemma}

\newtheorem{proposition}{Proposition}

\begin{document}

\title{{A Model for Randomized Resource Allocation \\in  Decentralized Wireless Networks}}
\author{ Kamyar Moshksar, Alireza Bayesteh and Amir K. Khandani \\
\small Coding \& Signal Transmission Laboratory (www.cst.uwaterloo.ca)\\
Dept. of Elec. and Comp. Eng., University of Waterloo\\ Waterloo, ON, Canada, N2L 3G1 \\
Tel: 519-725-7338, Fax: 519-888-4338\\e-mail: \{kmoshksa, alireza,
khandani\}@cst.uwaterloo.ca} \maketitle
\begin{abstract}
In this paper, we consider a decentralized  wireless communication
network with a fixed number $u$ of frequency sub-bands to be shared
among $N$ transmitter-receiver pairs. It is assumed that the number
of active users is a random variable with a given probability mass
function. Moreover, users are unaware of each other's codebooks and
hence, no multiuser detection is possible. We propose a randomized
Frequency Hopping (FH) scheme in which each transmitter randomly
hops over a subset of $u$ sub-bands from transmission to
transmission. Assuming all users transmit Gaussian signals, the
distribution of the noise plus interference is mixed Gaussian, which
makes calculation of the mutual information between the transmitted
and received signals of each user intractable.
 We derive lower and upper bounds
on the mutual information of each user and demonstrate that, for
large Signal-to-Noise Ratio (SNR) values, the two bounds coincide.
This observation enables us to compute the sum multiplexing
gain of the system and obtain the optimum hopping strategy for
maximizing this quantity. We compare the performance of the  FH
system with that of the Frequency Division (FD) system in terms of
the following performance measures: average sum multiplexing
gain $(\eta^{(1)})$, average minimum multiplexing gain per user
$(\eta^{(2)})$, minimum nonzero multiplexing gain per user
$(\eta^{(3)})$ and service capability ($\eta^{(4)}$). We show that
(depending on the probability mass function of the number of active
users) the FH system can offer a significant improvement in terms of
$\eta^{(1)}$, $\eta^{(2)}$ and $\eta^{(4)}$ (implying a more
efficient usage of the spectrum). It is also shown that
$\frac{1}{e}\leq\frac{\eta_{\mathrm{FH}}^{(3)}}{\eta_{\mathrm{FD}}^{(3)}}\leq
1$, i.e., the loss incurred in terms of $\eta^{(3)}$ is not more
than $\frac{1}{e}$.
\end{abstract}
\vskip 3cm
\begin{keywords}
Randomized Signaling, Frequency Hopping, Spectrum Sharing, Decentralized Networks, Mixed Gaussian Interference, Multiplexing Gain.
\end{keywords}
\section{Introduction}
\subsection{Motivation}
Increasing  demand for wireless applications on one hand, and the
limited available resources on the other hand, provoke more
efficient usage of such resources. Due to its significance, many
researchers have addressed the problem of resource allocation in
wireless networks.  One major challenge in wireless networks is the
destructive effect of multi-user interference, which degrades the
performance when multiple users share the spectrum. As such, an
efficient and low complexity resource allocation scheme that
maximizes the quality of service while mitigating the impact of the
multi-user interference is desirable.  The existing resource
allocation schemes are either \textit{centralized}, i.e., a central
controller manages the resources, or \textit{decentralized}, where
resource allocation is performed locally at each node. Due to the
complexity of adapting the centralized schemes to the network
structure (e.g. number of active users), these schemes are usually
designed for a fixed network structure. This makes inefficient
usage of resources because, in most cases, the number of active users may
be considerably less than the value assumed in the design process. On
the other hand, most of the decentralized resource allocation
schemes suffer from the complexity, either in the algorithm (e.g.
game-theoretic approaches involving iterative methods) or in the
hardware (e.g. cognitive radio). Therefore, it is of interest to
devise an efficient and low-complexity decentralized resource
allocation scheme, which is the main goal of this paper.

\subsection{Related Works}

\subsubsection{Centeralized Schemes}
In recent years, many centralized power and spectrum allocation
schemes have been studied in cellular and multihop wireless networks
\cite{ KumaranITWC0505,
ElBattITWC0104,HollidayISIT2004,WassermanITWC0603, HanITC0805,
KatzelaIPC0696, KianiISIT2006, LiangITIT1007}. Clearly, centralized
schemes perform better than the decentralized (distributed)
approaches, while requiring extensive knowledge of the network
configuration. In particular, when the number of nodes is large,
deploying such centralized schemes may not be practically feasible.

Traditional wireless systems aimed to avoid the interference among
users by using orthogonal transmission schemes. The most common
example is the Frequency Division (FD) system, in which different
users transmit over disjoint frequency sub-bands. The assignment of
frequency sub-bands is usually performed by a central controller.
Despite its simplicity, FD is shown to achieve the highest
throughput in certain scenarios. In particular, \cite{1} proves that
in a wireless network where interference is treated as noise
(no multi-user detection is performed), if the crossover gains are
sufficiently larger than the forward gains, FD is Pareto-rate-optimal. Due to
practical considerations, such FD systems usually rely on a fixed
number of frequency sub-bands. Hence, if the number of users
changes, the system is not guaranteed to offer the best possible
spectral efficiency  because, most of the time, the majority of the
potential users may be inactive.

\subsubsection{Decentralized Schemes}
In decentralized schemes, decisions concerning network resources
are made by individual nodes based on their local information.
Most of decentralized schemes reported in the literature rely on
either \textit{game-theoretic} approaches or \textit{cognitive
radios}. Cognitive radios \cite{2} have the ability to sense the
unoccupied portion of the available spectrum and use this
information in resource allocation. Fundamental limits of wireless
networks with cognitive radios are studied in \cite{3,4,5,7}.
Although cognitive radios avoid the use of a central controller,
they require sophisticated detection techniques for sensing the
spectrum holes and dynamic frequency assignment, which add to the
overall system complexity \cite{8}. Noting the above points, it is
desirable to have a decentralized frequency sharing strategy
without the need for cognitive radios, which allows the users to
coexist while utilizing the spectrum efficiently and fairly.

Being a standard technique in spread spectrum communications and due
to its interference avoidance nature, hopping is the
simplest spectrum sharing method to use in decentralized networks.
As different users typically have no prior information about the
codebooks of the other users, the most efficient method is avoiding
interference by choosing unused channels. As mentioned earlier,
searching the spectrum to find spectrum holes is not an easy task
due to the dynamic spectrum usage. As such, Frequency Hopping (FH) is a realization of a
transmission scheme without sensing, while avoiding the collisions
as much as possible.  Frequency Hopping is one of the standard
signaling schemes~\cite{15} adopted in ad-hoc networks. In short
range scenarios, bluetooth systems \cite{19,20,21} are the most
popular examples of a Wireless Personal Area Network (WPAN). Using
FH over the unlicensed ISM band, a bluetooth system provides robust
communication to unpredictable sources of interference. A
modification of Frequency Hopping, called Dynamic Frequency Hopping
(DFH), selects the hopping pattern based on interference
measurements in order to avoid dominant interferers. The performance
of a DFH scheme when applied to a cellular system is assessed in
\cite{22,23,24}.

In \cite{Jindal}, the authors consider the problem of bandwidth
partitioning in a decentralized wireless  network where different transmitters are connected to
different receivers through channels with similar path loss
exponent. Assuming the transmitters are scattered over the two
dimensional plane according to a Poisson point process, a fixed
bandwidth is partitioned into a certain number of sub-bands such
that the so-called transmission intensity in the network is
maximized while the probability of outage per user is below a
certain threshold. The transmission strategy is based on choosing
one sub-band randomly per transmission, which is a special case of
FH.

Frequency hopping is also proposed in \cite{7} in the context of
cognitive radios, where each cognitive transmitter selects a
frequency sub-band but quits transmitting if the sub-band is already
occupied by a primary user.

Recently, Orthogonal Frequency Division Multiplexing (OFDM) has been
considered as a promising technique in many wireless technologies.
OFDM partitions a wide-band channel to a group of narrow-band
orthogonal sub-channels. The popularity of OFDM motivates us to
consider a Frequency Hopping scheme operating over $u$ narrow-band
orthogonal frequency sub-bands. We note that the results of the
paper are valid in a general setup where  hopping is performed over
an arbitrary orthogonal basis. To make the presentation as simple as
possible, we take the orthogonal basis in frequency, which can be
realized in practice using OFDM systems.

\subsection{Contribution}
In this paper, we consider a decentralized wireless communication
network with a fixed number $u$ of frequency sub-bands to be shared
among $N$ transmitter-receiver pairs. It is assumed that the number
of active users is a random variable with a given probability mass
function. Moreover, users are unaware of each other's codebooks,
and hence, no multiuser detection is possible.  We propose a
randomized Frequency Hopping scheme in which the $i^{th}$
transmitter randomly hops over $v_{i}$ out of $u$ sub-bands from
transmission to transmission. Assuming i.i.d. Gaussian signals are
transmitted over the chosen sub-bands, the distribution of the noise
plus interference becomes mixed Gaussian, which makes the calculation
of the achievable rate complicated. The main contributions of the
paper are:
 \begin{itemize}
\item We derive lower and upper bounds
on the mutual information between the transmitted and received
signals of each user and demonstrate that, for large SNR values, the
two bounds coincide. Thereafter,  we are able to show that the achievable rate of
the $i^{th}$ user scales like $\frac{v_{i}}{2} \prod_{\substack{j=1\\j\neq i}}^{N}\left(1-\frac{v_{j}}{u}\right) \log\mathsf{SNR}$.

\item We show that each transmitter only needs the knowledge
of the number of active users in the network, the forward channel
gain and the maximum interference at its desired receiver to
regulate its transmission rate. Knowing these quantities, we
demonstrate how the $i^{th}$ user can achieve a multiplexing gain of
$\frac{v_{i}}{2} \prod_{\substack{j=1\\j\neq
i}}^{N}\left(1-\frac{v_{j}}{u}\right)$.

\item
We obtain the optimum design parameters $\{v_{i}\}_{i=1}^{N}$ in
order to maximize various performance measures.

\item
We compare the performance of the FH with that of the Frequency
Division in terms of the following performance measures:
average sum multiplexing gain $(\eta^{(1)})$, average minimum 
multiplexing gain per user $(\eta^{(2)})$, minimum nonzero
multiplexing gain per user $(\eta^{(3)})$ and service capability
($\eta^{(4)}$). We show that (depending on the probability mass
function of the number of active users) the FH system can offer a
significant improvement in terms of $\eta^{(1)}$, $\eta^{(2)}$, and
$\eta^{(4)}$ (implying a more efficient usage of the spectrum). It
is also shown that
$\frac{1}{e}\leq\frac{\eta_{\mathrm{FH}}^{(3)}}{\eta_{\mathrm{FD}}^{(3)}}\leq
1$, i.e., the loss incurred in terms of $\eta^{(3)}$ is not more
than $\frac{1}{e}$.

\end{itemize}

The paper outline is as follows. The system model is given in section
II. Section III offers an analysis of the achievable rates. Upper
bounds and lower bounds on the achievable rates of users are
presented in this section. In section IV, based on the results in
section III, we discuss how users in the FH system can fairly share
the spectrum. Comparison with the FD scheme in terms of different
performance measures is discussed in section V. Section VI offers a
comparison between two versions of the proposed FH, i.e., the robust
frequency hopping and adaptive frequency hopping. Finally, section
VII states the concluding remarks.

\subsection{Notation}
Throughout the paper, we use the notation $\mathrm{E} \{.\}$ for the
expectation operator. $\mathrm{Pr}\{\mathcal{E}\}$ denotes the
probability of an event $\mathcal{E}$, $\mathbb{1}(\mathcal{E})$ the
indicator function of an event $\mathcal{E}$ and $p_{X}(.)$  the
probability density function (PDF) of a random variable $X$. Also,
$\mathrm{I} (X;Y)$ denotes the mutual information between random
variables $X$ and $Y$ and $\mathrm{h} (X)$ denotes the differential entropy of
a continuous random variable $X$. Finally, the notation $f(\gamma)\sim
g(\gamma)$ implies
$\lim_{\gamma\rightarrow\infty}\frac{f(\gamma)}{g(\gamma)}=1$.

\section{System Model and Assumptions}
We consider a wireless network with $N$ users\footnote{Each user
consists of a transmitter-receiver pair.} operating over a spectrum
consisting of $u$ orthogonal sub-bands. The number of active users
is assumed to be a random variable with a given distribution, however, it is fixed during the whole transmission once it is set first. The
transmission blocks of each user comprise of an arbitrarily large
number of transmission slots. We remark that the results of this paper are valid regardless of having block synchronization among the users, however, we assume synchronization at the symbol level. It is assumed that the $i^{th}$ user
exploits $v_{i}(\leq u)$ out of the $u$  sub-bands in each
transmission slot and hops randomly to another set of $v_{i}$
frequency sub-bands in the next transmission slot.  This user
transmits independent real Gaussian signals of variance
$\frac{P}{v_{i}}$ over the chosen sub-bands, in which $P$ denotes
the total average power for each transmitter. Each receiver is
assumed to know the hopping pattern of its affiliated transmitter.
It is assumed that the users are not aware of each other's
codebooks and hence, no multiuser detection or interference
cancelation is possible at the receiver sides. The
static and non-frequency selective channel gain of the link
connecting the $i^{th}$ transmitter to the $j^{th}$ receiver is
shown by $h_{i,j}$. As it will be shown in (\ref{hoolooo}), the only
information each transmitter needs  in order to regulate its
transmission rate (focusing on the achieved multiplexing gain) is
its forward channel gain, the maximum interference level at its
associated receiver and the number of active users in the network.
This information can be obtained at the receiver side by
investigating the interference PDF and provided to the
corresponding receiver via a feedback link.

As all users hop over different portions of the spectrum from
transmission slot to transmission slot, no user is assumed to be capable of
tracking the instantaneous interference level. This assumption makes
the interference plus noise PDF at the receiver side of each user 
be mixed Gaussian. In fact, depending on different choices the other
users make to select the frequency sub-bands and values of the
crossover gains, the interference on each frequency sub-band at any
given receiver can have up to $2^{N-1}$ power levels. The vector
consisting of the received signals on the frequency sub-bands at the
$i^{th}$ receiver in a typical transmission slot is
\begin{equation}
\vec{Y}_{i}=h_{i,i}\vec{X}_{i}+\vec{Z}_{i},
\end{equation}
where $\vec{X}_{i}$ is the $u\times 1$ transmitted vector and
$\vec{Z_{i}}$ is the noise plus interference vector at the receiver
side of the $i^{th}$ user. Due to the fact that each transmitter
hops randomly from slot to slot, one may write $p_{\vec{X}_{i}}(.)$
as
\begin{eqnarray}
p_{\vec{X}_{i}}(\vec{x})=\sum_{C\in \mathfrak{C}}\frac{1}{{u\choose v_{i}}}g_{u}(\vec{x},C),
\end{eqnarray}
which corresponds to the mixed Gaussian distribution. In the above
equation, $g_{u}(\vec{x},C)$ denotes the PDF of a zero-mean $u\times
1$ jointly Gaussian vector of covariance matrix $C$ and the set
$\mathfrak{C}$ includes all $u\times u$ diagonal matrices in which
$v_{i}$ out of the $u$ diagonal elements are $\frac{P}{v_{i}}$ and
the rest are zero.  Denoting the noise plus interference on the
$j^{th}$ sub-band at the receiver side of the $i^{th}$ user by
$Z_{i,j}$ (the $j^{th}$ component of $\vec{Z}_{i}$), it is clear
that $p_{Z_{i,j}}(.)$ is not dependent on $j$. This is due to the
fact that the crossover gains are not frequency selective and
there is no particular interest to a specific frequency sub-band by
any user. We assume there are $L_{i}+1$ ($L_{i}\leq 2^{N-1}-1$)
possible non-zero power levels for $Z_{i,j}$, say
$\{\sigma^{2}_{i,l}\}_{l=0}^{L_{i}}$. Denoting the occurrence
probability of $\sigma^{2}_{i,l}$ by $a_{i,l}$, $p_{Z_{i,j}}(.)$
identifies a mixed Gaussian PDF as
\begin{equation}
\label{hf}
p_{Z_{i,j}}(z)=\sum_{l=0}^{L_{i}}\frac{a_{i,l}}{\sqrt{2\pi}\sigma_{i,l}}\exp \left(-\frac{z^2}{2\sigma_{i,l}^{2}} \right),
\end{equation}
where  $\sigma^{2}=\sigma_{i,0}^{2}< \sigma_{i,1}^{2}<
\sigma_{i,2}^{2}<...<\sigma_{i,L_{i}}^{2}$ ($\sigma^{2}$ is the
ambient noise power). We notice that for each $l\geq 0$, there
exists a $c_{i,l} \geq 0$ such that
$\sigma_{i,l}^{2}=\sigma^{2}+c_{i,l}P$ where
$0=c_{i,0}<c_{i,1}<c_{i,2}<...<c_{i,L_{i}}$. One may write $Z_{i,j}=
\sum_{\substack{k=1\\k\neq
i}}^{N}\epsilon_{k,j}h_{k,i}X_{k,j}+\nu_{i,j}$ where $\epsilon_{k,j}$ is a Bernoulli random variable showing if the
$k^{th}$ user has utilized the $j^{th}$ sub-band, $X_{k,j}$ is
the signal of the $k^{th}$ user sent on the $j^{th}$ sub-band (assuming it has utilized the $j^{th}$ sub-band),
 and $\nu_{i,j}$ is
the ambient noise which is a zero-mean Gaussian random variable with
variance $\sigma^{2}$. The ratio $\frac{P}{\sigma^{2}}$ is taken as
a measure of SNR and is denoted by $\gamma$ throughout the paper.

\section{Analysis of the Achievable Rate}
Let us denote the achievable rate of the $i^{th}$ user by
$\mathscr{R}_{i}$. It can be observed that the communication channel
of this user is a channel with state $S_{i}$, the hopping pattern of
the $i^{th}$ user, which is independently changing over different
transmission slots, and is known to both the transmitter and the receiver.
The achievable rate of such a channel is given by
\begin{equation}
\mathscr{R}_{i}=\mathrm{I}(\vec{X}_{i};\vec{Y}_{i}\vert S_{i})=\sum_{s_{i}\in \mathfrak{S}_{i}}\Pr(S_{i}=s_{i})\mathrm{I}(\vec{X}_{i};\vec{Y}_{i}\vert S_{i}=s_{i}),
\end{equation}
where $\mathrm{I}(\vec{X}_{i};\vec{Y}_{i}\vert S_{i}=s_{i})$ is the
mutual information between $\vec{X_{i}}$ and $\vec{Y_{i}}$ for the
specific sub-band selection corresponding to $S_{i}=s_{i}$. The set
$\mathfrak{S}_{i}$ denotes all possible selections of $v_{i}$ out of
the $u$ sub-bands. As $p_{\vec{Z_{i}}}(.)$ is a symmetric density
function, meaning all its components have the same PDF given in
(\ref{hf}), we deduce that $\mathrm{I}(\vec{X}_{i};\vec{Y}_{i}\vert
S_{i}=s_{i})$ is independent of $s_{i}$. Therefore, to calculate
$\mathscr{R}_{i}$, we may assume any specific sub-band selection for
the $i^{th}$ user in $\mathfrak{S}_{i}$, say the first $v_{i}$
sub-bands. Denoting this specific state by $s^{*}_{i}$, we get
\begin{equation}
\mathscr{R}_{i}=\mathrm{I}(\vec{X}_{i};\vec{Y}_{i}\vert S_{i}=s^{*}_{i}).
\end{equation}
In this case, we denote $\vec{Y}_{i}$ and $\vec{X}_{i}$ by
$\vec{Y}_{i}(s^{*}_{i})$ and $\vec{X}_{i}(s^{*}_{i})$, respectively.
Obviously, we have
\begin{equation}
\label{lala}
\mathscr{R}_{i}=\mathrm{I}(\vec{X}_{i}(s^{*}_{i});\vec{Y}_{i}(s^{*}_{i}))=\mathrm{h}(\vec{Y}_{i}(s^{*}_{i}))-\mathrm{h}(\vec{Z}_{i}).
\end{equation}
Because $\vec{Y}_{i}(s^{*}_{i})$ and $\vec{Z}_{i}$ have mixed Gaussian
distribution, there is no closed-form expression for the
differential entropy of these vectors. As such, we provide an upper
bound and a lower bound on the achievable rate of each user in the
following subsections and show that these bounds converge in the
asymptotic high SNR regime.

 \subsection{Upper Bound on The Achievable Rates}
In this section, we develop an upper bound
$\mathscr{R}_{i,\mathrm{ub}}$ on the achievable rate of the $i^{th}$
user that is tight enough to ensure that
$\mathscr{R}_{i,\mathrm{ub}}-\mathscr{R}_{i}$ does not increase
unboundedly as SNR increases. The idea behind this upper bound is
the convexity of $\mathscr{R}_{i}$ in terms of
$p_{\vec{Y}_{i}(s^{*}_{i})\vert \vec{X}_{i}(s^{*}_{i})}(.|.)$.

\begin{thm}
There exists an upper bound on the achievable rate of the $i^{th}$ user given by
 \begin{equation}
\label{loopy}
\mathscr{R}_{i,\mathrm{ub}}=\frac{1}{2}v_{i}\prod_{\substack{k=1\\k\neq i}}^{N}\left(1-\frac{v_{k}}{u}\right)\log\left(1+\frac{\vert h_{i,i}\vert^{2}\gamma}{v_{i}}\right)+\widetilde{\mathscr{R}}_{i,\mathrm{ub}}
\end{equation}
where
$\lim_{\gamma\rightarrow\infty}\widetilde{\mathscr{R}}_{i,\mathrm{ub}}<\infty$.
In particular,
$\mathscr{R}_{i,\mathrm{ub}}\sim\frac{1}{2}v_{i}\prod_{\substack{k=1\\k\neq
i}}^{N}\left(1-\frac{v_{k}}{u}\right)\log\gamma$.
\end{thm}
\begin{proof}
Let $\vec{W}_{i}$ be the $u\times 1$ interference vector where its
$j^{th}$ component $W_{i,j}$ is a random variable showing the
interference term on the $j^{th}$ frequency sub-band at the receiver
of the $i^{th}$ user. We have $W_{i,j}=\sum_{\substack{k=1\\k\neq
i}}^{N}\epsilon_{k,j}h_{k,i}X_{k,j}$. Clearly, $\vec{W}_{i}$ is a
mixed Gaussian random vector where the Gaussian components in its PDF represent
different choices the other users make in selecting their sub-bands.
In fact, we have
$p_{\vec{W}_{i}}(\vec{w})=\frac{1}{M_{i}}\sum_{m=1}^{M_{i}}g_{u}(\vec{w},D_{i,m})$,
where $M_{i}=\prod_{k\neq i}{u\choose v_{k}}$ and
$D_{i,m}=\mathrm{diag}(d_{i,m}^{(1)},\cdots,d_{i,m}^{(u)})$, in
which $d_{i,m}^{(j)} = \sum_{\substack{k=1\\k\neq i}}^{N}
\epsilon_{k,j}^2 |h_{k,i}|^{2} \frac{P}{v_k}$ denotes the variance
of $W_{i,j}$ for the $m^{th}$ realization
of $\{\epsilon_{k,j}\}_{k\neq i}$ out of $M_i$ possible realizations\footnote{Note that as each
user transmits independent Gaussian signals over its chosen
sub-bands, the matrices $\{D_{i,m}\}_{m=1}^{M_{i}}$ are diagonal.}. If the probability density
function of the interference vector consisted only of
$g_{u}(\vec{w},D_{i,m})$, the forward link of the $i^{th}$ user
would be converted into an additive Gaussian channel. The achievable
rate of such a virtual channel is given by
\begin{eqnarray}
\mathscr{R}_{i,m}&=&\frac{1}{2}\log\frac{\det\left(\mathrm{Cov}(\vec{X}_{i}(s^{*}_{i}))+D_{i,m}+\sigma^{2}I_{u}\right)}{\det\left(D_{i,m}+\sigma^{2}I_{u}\right)} \notag\\
&=& \frac{1}{2}\log\frac{\prod_{j=1}^{v_{i}}\left(\frac{\vert h_{i,i}\vert^{2}P}{v_{i}}+d_{i,m}^{(j)}+\sigma^{2}\right)}{\prod_{j=1}^{v_{i}}(d_{i,m}^{(j)}+\sigma^{2})} \notag\\&=&\frac{1}{2}\sum_{j=1}^{v_{i}}\log\bigg(1+\frac{\vert h_{i,i}\vert^{2}P}{v_{i}(d_{i,m}^{(j)}+\sigma^{2})}\bigg).
\end{eqnarray}
One may also state this as follows. Let $T_{i,m}\triangleq\left\{ j
: 1\leq j\leq v_{i}, d_{i,m}^{(j)}=0\right\}$. Then,
\begin{equation}
\mathscr{R}_{i,m}=\frac{\vert T_{i,m}\vert}{2}\log\bigg(1+\frac{\vert h_{i,i}\vert^{2}\gamma}{v_{i}}\bigg)+\widetilde{\mathscr{R}}_{i,m},
\end{equation}
where
\begin{equation}
\widetilde{\mathscr{R}}_{i,m}=\frac{1}{2}\sum_{1\leq j\leq v_{i}: d_{i,m}^{(j)}\neq 0}\log\bigg(1+\frac{\vert h_{i,i}\vert^{2}P}{v_{i}(d_{i,m}^{(j)}+\sigma^{2})}\bigg)
\end{equation}
and $\vert T_{i,m}\vert$ denotes the cardinality of the set
$T_{i,m}$. As each nonzero $d_{i,m}^{(j)}$ is proportional to $P$,
it is clear that
$\lim_{\gamma\rightarrow\infty}\widetilde{\mathscr{R}}_{i,m}<\infty$.
We know that $\mathscr{R}_{i}$ is convex in terms of
$p_{\vec{Y}_{i}(s_{i}^{*})\vert
\vec{X}_{i}(s_{i}^{*})}(\vec{y}\vert\vec{x})=p_{\vec{Z}_{i}}(\vec{y}-h_{i,i}\vec{x})$
\cite{53}. Noting this and the fact that
$p_{\vec{Z}_{i}}(\vec{z})=\frac{1}{M_{i}}\sum_{m=1}^{M_{i}}g_{u}(\vec{z},D_{i,m}+\sigma^{2}I_{u})$,
we have
\begin{equation}
\label{lal}
\mathscr{R}_{i}\leq \frac{1}{M_{i}}\sum_{m=1}^{M_{i}}\mathscr{R}_{i,m}=\left(\frac{1}{M_{i}}\sum_{m=1}^{M_{i}}\vert T_{i,m}\vert\right)\frac{1}{2}\log\bigg(1+\frac{\vert h_{i,i}\vert^{2}\gamma}{v_{i}}\bigg)+\widetilde{\mathscr{R}}_{i,\mathrm{ub}},
\end{equation}
where
$\widetilde{\mathscr{R}}_{i,\mathrm{ub}}=\frac{1}{M_{i}}\sum_{m=1}^{M_{i}}\widetilde{\mathscr{R}}_{i,m}$.
Clearly, as each $\widetilde{\mathscr{R}}_{i,m}$ saturates by
increasing $\gamma$, one has
$\lim_{\gamma\rightarrow\infty}\widetilde{\mathscr{R}}_{i,\mathrm{ub}}<\infty$.
The following Lemma offers an explicit expression for
$\frac{1}{M_{i}}\sum_{m=1}^{M_{i}}\vert T_{i,m}\vert$.
\begin{lem}
\begin{equation}\label{nhjbv}\frac{1}{M_{i}}\sum_{m=1}^{M_{i}}\vert T_{i,m}\vert=v_{i}\prod_{\substack{k=1\\k\neq i}}^{N}\left(1-\frac{v_{k}}{u}\right).\end{equation}
\end{lem}
\begin{proof}
Defining $A_{i,j}\triangleq\left\{ m: 1\leq m\leq M_{i}, \vert
T_{i,m}\vert=j\right\}$ for each $1\leq i\leq N$ and $1\leq j\leq
v_{i}$, one may express the left-hand side of (\ref{nhjbv}) as
\begin{equation}
\label{lki} \frac{1}{M_{i}}\sum_{m=1}^{M_{i}}\vert
T_{i,m}\vert=\frac{1}{M_{i}}\sum_{j=1}^{v_{i}}j\vert
A_{i,j}\vert.\end{equation}
Let $F_{i}$ be a random variable showing
the number of interference-free sub-bands among the $v_{i}$
sub-bands selected by the $i^{th}$ user. Using (\ref{lki}) and noting that
$\Pr\{F_{i}=j\}=\frac{\vert A_{i,j}\vert}{M_{i}}$, 
\begin{equation}\frac{1}{M_{i}}\sum_{m=1}^{M_{i}}\vert T_{i,m}\vert=\sum_{j=1}^{v_{i}}j\Pr\{F_{i}=j\}=\mathrm{E}\{ F_{i}\}.\end{equation}
Let us define
\begin{equation}F_{i,j}\triangleq\left\{ \begin{array}{cc}
1 &  W_{i,j}=0 \\
0 &  W_{i,j}\neq 0 \end{array}\right.\end{equation}
for any $1\leq i\leq N$ and $1\leq j\leq v_{i}$.
Obviously, $F_{i}=\sum_{j=1}^{v_{i}}F_{i,j}$. As such,
\begin{equation}\mathrm{E}\lbrace F_{i}\rbrace=\sum_{j=1}^{v_{i}}\mathrm{E}\{ F_{i,j}\}=\sum_{j=1}^{v_{i}}\Pr\{W_{i,j}=0\}.\end{equation}
Since $\Pr\{\epsilon_{k,j}=1\}=\frac{v_{k}}{u}$,
\begin{eqnarray}\Pr\{W_{i,j}=0\}=\Pr\{\textrm{$Z_{i,j}$ contains no interference}\}=\prod_{\substack{k=1\\k\neq i}}^{N}\Pr\{\epsilon_{k,j}=0\}=\prod_{\substack{k=1\\k\neq i}}^{N}\left(1-\frac{ v_{i}}{u}\right).
\end{eqnarray}  This yields
\begin{equation}\mathrm{E}\{ F_{i}\}=v_{i}\prod_{\substack{k=1\\k\neq i}}^{N}\left(1-\frac{v_{k}}{u}\right),\end{equation}
which completes the proof of Lemma 1.
\end{proof}
Based on (\ref{lal}) and Lemma 1, the proof of Theorem 1 is complete.
\end{proof}

\subsection{Lower Bound on the Achievable Rates}
In this section, we derive a lower bound on the achievable rates of
users. The idea behind deriving this lower bound is to invoke the
classical entropy power inequality (EPI). As we will see, this
initial lower bound is not in a closed form as it depends on the
differential entropy of a mixed Gaussian random variable. In
appendix A,  we obtain an appropriate upper bound on such an entropy
which leads to the final lower bound on $\mathscr{R}_{i}$.
\begin{thm}
There exists a lower bound $\mathscr{R}_{i,\mathrm{lb}}$ on the
achievable rate of the $i^{th}$ user which can be written as
\begin{equation}\mathscr{R}_{i,\mathrm{lb}} = \frac{1}{2}v_{i}\prod_{\substack{k=1\\k\neq i}}^{N}\left(1-\frac{v_{k}}{u}\right)\log\gamma + \widetilde{\mathscr{R}}_{i,\mathrm{lb}},\end{equation}
such that $\lim_{\gamma \to \infty}
\widetilde{\mathscr{R}}_{i,\mathrm{lb}} < \infty$.  In particular,
$\mathscr{R}_{i,\mathrm{lb}}\sim\frac{1}{2}v_{i}\prod_{\substack{k=1\\k\neq
i}}^{N}\left(1-\frac{v_{k}}{u}\right)\log\gamma$.
\end{thm}
\begin{proof}
We define $\vec{X}'_{i}$ to be the $v_{i}\times 1$ signal vector
corresponding to the first $v_{i}$ elements of
$\vec{X}_{i}(s^{*}_{i})$. Clearly, $\vec{X}'_{i}$ is a Gaussian
vector with covariance matrix $\frac{P}{v_i} I_{v_i}$. Let
$\vec{Y}'_{i}=h_{i,i}\vec{X}'_{i}+\vec{Z}'_{i}$ where $\vec{Z}'_{i}$
is the noise plus interference vector at the receiver side of the
$i^{th}$ user on the first $v_{i}$ sub-bands. Using entropy power
inequality, we have
\begin{equation}
2^{\frac{2}{v_{i}}\mathrm{h}(\vec{Y}'_{i})}\geq 2^{\frac{2}{v_{i}}\mathrm{h}(h_{i,i}\vec{X}'_{i})}+2^{\frac{2}{v_{i}}\mathrm{h}(\vec{Z}'_{i})}.
\end{equation}
Dividing both sides by $2^{\frac{2}{v_i}\mathrm{h}(\vec{Z}'_{i})}$, we get
\begin{equation}
\label{piu}
\mathrm{h}(\vec{Y}'_{i})-\mathrm{h}(\vec{Z}'_{i})\geq \frac{v_{i}}{2}\log\bigg(2^{\frac{2}{v_{i}}\left(\mathrm{h}(h_{i,i}\vec{X}'_{i})-\mathrm{h}(\vec{Z}'_{i})\right)}+1\bigg).
\end{equation}
On the other hand, since $\vec{Y}'_{i}$ is a subvector of $\vec{Y}_{i}(s^{*}_{i})$, we have
\begin{equation}
\label{qiu}
\mathscr{R}_{i}=\mathrm{I}(\vec{X}_{i}(s^{*}_{i});\vec{Y}_{i}(s^{*}_{i}))\geq \mathrm{I}(\vec{X}'_{i};\vec{Y}'_{i})=\mathrm{h}(\vec{Y}'_{i})-\mathrm{h}(\vec{Z}'_{i}).
\end{equation}
Comparing (\ref{piu}) and (\ref{qiu}) yields
\begin{equation}
\label{uuu}
\mathscr{R}_{i}\geq \frac{v_{i}}{2}\log\bigg(2^{\frac{2}{v_{i}}\left(\mathrm{h}(h_{i,i}\vec{X}'_{i})-\mathrm{h}(\vec{Z}'_{i})\right)}+1\bigg).
\end{equation}
Clearly,
$\mathrm{h}(h_{i,i}\vec{X}'_{i})=\frac{v_{i}}{2}\log\left(2\pi
e\frac{\vert h_{i,i}\vert^{2} P}{v_{i}}\right)$. As $\vec{Z}'_{i}$
is a mixed Gaussian random vector, there is no closed-form formula
for $\mathrm{h}(\vec{Z}'_{i})$. Hence, we have to find an
appropriate upper bound on $\mathrm{h}(\vec{Z}'_{i})$ to further
simplify (\ref{uuu}). Using the chain rule for the differential
entropy, we obtain
\begin{equation}
\label{hjk}
\mathrm{h}(\vec{Z}'_{i})\leq\sum_{j=1}^{v_{i}}\mathrm{h}(Z_{i,j}).
\end{equation}
Recalling the definitions of $\{a_{i,l}\}_{l=0}^{L_{i}}$ and
$\{c_{i,l}\}_{l=0}^{L_{i}}$ in the system model, the following Lemma
yields an upper bound on $\mathrm{h}(Z_{i,j})$ for each $1\leq j\leq
v_{i}$.
\begin{lem}
For every $1\leq j\leq v_{i}$ and for all values of $\gamma$, there exists an upper bound on $\mathrm{h}(Z_{i,j})$ given by
\begin{equation}
\label{uu-u}
\mathrm{h}(Z_{i,j})\leq \frac{1-a_{i,0}}{2}\log(c_{i,L_{i}}\gamma+1)+\log(\sqrt{2\pi e}\sigma)+\mathscr{H}_{i}
\end{equation}
where $\mathscr{H}_{i}\triangleq-\sum_{l=0}^{L_{i}}a_{i,l}\log a_{i,l}$  is the discrete entropy of $\{a_{i,l}\}_{l=0}^{L_{i}}$.
\end{lem}
\begin{proof}
See Appendix A.
\end{proof}
By (\ref{uu-u}), (\ref{hjk}) and (\ref{uuu}),
\begin{eqnarray}
\label{uuku}
\mathscr{R}_{i}\geq \mathscr{R}_{i,\mathrm{lb}}&\triangleq&\frac{v_{i}}{2}\log\bigg(\frac{2^{-2\mathscr{H}_{i}}\vert h_{i,i}\vert^{2}\gamma}{v_{i}(c_{i,L_{i}}\gamma+1)^{1-a_{i,0}}}+1\bigg) \notag\\
&=& \frac{1}{2} v_{i} a_{i,0}\log \gamma + \frac{v_{i}}{2} \log \left( \frac{2^{-2\mathscr{H}_{i}}\vert h_{i,i}\vert^{2}}{v_{i}(c_{i,L_{i}}+\gamma^{-1})^{1-a_{i,0}}}+\gamma^{-a_{i,0}}\right).
\end{eqnarray}
Defining $\widetilde{\mathscr{R}}_{i,\mathrm{lb}} \triangleq
\frac{v_{i}}{2} \log \left( \frac{2^{-2\mathscr{H}_{i}}\vert
h_{i,i}\vert^{2}}{v_{i}(c_{i,L_{i}}+\gamma^{-1})^{1-a_{i,0}}}+\gamma^{-a_{i,0}}\right)$,
we note that $\lim_{\gamma \to \infty}
\widetilde{\mathscr{R}}_{i,\mathrm{lb}} < \infty $. Combining this
with the fact that $a_{i,0} = \prod_{\substack{k=1\\k\neq
i}}^{N}\left(1-\frac{v_{k}}{u}\right)$ completes the proof of
Theorem 2.
\end{proof}
In \cite{kami-1}, we address another approach to propose a lower
bound on the achievable rate of the $i^{th}$ user with the same SNR
scaling as $\mathscr{R}_{i,\mathrm{lb}}$.

One may consider the following generalization of the FH scheme.
Let us assume that the users are not restricted to choose a fixed
number of frequency sub-bands in each transmission slot. In fact, in
each transmission slot the number of selected sub-bands can be any
integer between $0$ and $u$, and the probability of choosing $v\in[0,u]\cap\mathbb{Z}$ sub-bands by the $i^{th}$ user is denoted by $\mu_{i,v}$.
Therefore, the $i^{th}$ user has two random generators. The first
random generator selects a number $0\leq v\leq u$ according to the
probability mass function $\{\mu_{i,v}\}_{v=0}^{u}$, while the other
generator selects $v$ sub-bands among the whole available $u$
sub-bands. This repeats independently from transmission slot to
transmission slot. Based on the arguments made in section II, the
achievable rate of the $i^{th}$ user can be written as
\begin{equation}
\label{hfd}
\mathscr{R}_{i}=\sum_{v=0}^{u}\mu_{i,v}\mathrm{I}(\vec{X}_{i}(s^{*}_{i,v});\vec{Y}_{i}(s^{*}_{i,v})),
\end{equation}
where $s^{*}_{i,v}$ denotes the state where the $i^{th}$ user
selects the first $v$ sub-bands. Clearly,
$\mathrm{I}(\vec{X}_{i}(s^{*}_{i,0});\vec{Y}_{i}(s^{*}_{i,0}))=0$
for any $1\leq i\leq N$. Furthermore, 
 \begin{eqnarray} \label{27}
 a_{i,0} &=& \Pr\left\{\textrm{A given component of $\vec{Z}_{i}$ contains no interference}\right\} \notag\\
 &=&\sum_{v_{1}=0}^{u}\cdots \sum_{v_{i-1}=0}^{u} \sum_{v_{i+1}=0}^{u} \cdots \sum_{v_{N}=0}^{u}\prod_{\substack{k=1\\k\neq i}}^{N}\mu_{k,v_{k}} \left(1-\frac{v_{k}}{u} \right) \notag\\
&=& \prod_{\substack {k=1\\k \neq i}}^{N}\sum_{v_{k}=0}^{u}\mu_{k,v_{k}} \left(1-\frac{v_{k}}{u} \right) \notag\\
 &=& \prod_{\substack {k=1\\k \neq i}}^{N}\sum_{v=0}^{u}\mu_{k,v} \left(1-\frac{v}{u}\right) \notag\\
&=& \prod_{\substack {k=1\\k \neq i}}^{N} \left(1-\frac{\bar{v}_k}{u}\right)
   \end{eqnarray}
where $\bar{v}_k\triangleq\sum_{v=0}^{u}v\mu_{k,v}$. Based on the results of
this section, we get
 \begin{equation}
 \label{poih}
 \mathrm{I}(\vec{X}_{i}(s^{*}_{i,v});\vec{Y}_{i}(s^{*}_{i,v}))\sim \frac{1}{2}va_{i,0}\log\gamma.
  \end{equation}
Using (\ref{hfd}), (\ref{27}) and (\ref{poih}) yields
\begin{eqnarray}
\label{mach}
 \mathscr{R}_i &\sim& \sum_{v=0}^{u} \frac{1}{2}\mu_{i,v} v \prod_{\substack {k=1\\k \neq i}}^{N}\left(1-\frac{\bar{v}_k}{u}\right) \log \gamma\notag\\
 &=& \frac{1}{2} \bar{v}_i \prod_{\substack {k=1\\k \neq i}}^{N} \left(1-\frac{\bar{v}_k}{u}\right) \log \gamma.
\end{eqnarray}
In fact, (\ref{mach}) demonstrates that the generalized FH scheme is equivalent to
the FH scheme through substituting $\{v_i\}_{i=1}^{N} $ by
$\{\bar{v}_i\}_{i=1}^N$. However, it is remarkable that in contrast
to the FH scheme in which $\{v_i\}_{i=1}^{N} $ are integer values,
in the generalized FH scheme  $\{\bar{v}_i\}_{i=1}^N$ are real
values. This provides more flexibility in system design. The above
observation motivates us to use the generalized scenario in the
sequel and we simply refer to it as the FH scheme. In this scheme,
the $i^{th}$ user has a parameter $\bar{v}_i$, which can be chosen to
be any real number in the interval $[0,u]$.

\section{System Design}
In this section, we find the optimum operation point  for the FH scheme. This requires finding
the optimum values of $\{\bar{v}_i\}_{i=1}^N$.  Based on the results
established in the previous section, there exist upper and lower
bounds on the achievable rate of each user that coincide in the
high SNR regime. As such, the achievable rate itself must be
asymptotically equivalent to each of these bounds, i.e.,
\begin{equation}
\label{uu}
\mathscr{R}_{i}\sim \frac{1}{2} \bar{v}_{i}\prod_{\substack{k=1\\k\neq i}}^{N}\left(1-\frac{\bar{v}_{k}}{u}\right)\log\gamma,
\end{equation}
where, based on the conclusion made at the end of section III, the
parameters $\{\bar{v}_{i}\}_{i=1}^{N}$ can be adjusted to be any
real number in the range $[0,u]$. By (\ref{uu}), the network sum-rate can be asymptotically written as
\begin{equation}
\sum_{i=1}^{N}\mathscr{R}_{i}\sim \mathsf{SMG}\left(\bar{v}_{1},\cdots,\bar{v}_{N}\right)\log\gamma,
\end{equation}
where
\begin{equation}
\mathsf{SMG}\left(\bar{v}_{1},\cdots,\bar{v}_{N}\right)\triangleq\sum_{i=1}^{N}
\frac{1}{2} \bar{v}_{i}\prod_{\substack{k=1\\k\neq
i}}^{N}\left(1-\frac{\bar{v}_{k}}{u}\right).\end{equation} We call
$\mathsf{SMG}\left(\bar{v}_{1},\cdots,\bar{v}_{N}\right)$
the \textit{sum multiplexing gain} of the system.
$\mathsf{SMG}\left(\bar{v}_{1},\cdots,\bar{v}_{N}\right)$
is a symmetric function of $(\bar{v}_{1},\cdots,\bar{v}_{N})$ and
has a saddle point at $\bar{v}_{i}=\frac{u}{N}$ for  $1\leq i\leq
N$. In a \emph{fair} FH system, it is required that $\bar{v}_{i}=v$ for
all $1\leq i\leq N$ where $v$ is any real number in the interval
$[0,u]$. Hence, we define
\begin{equation}
\label{yu}
\mathsf{SMG}(v,N)\triangleq\mathsf{SMG}\left(\bar{v}_{1},\cdots,\bar{v}_{N}\right)\Big|_{\forall i:\bar{v}_{i}=v}=\frac{N}{2} v\left(1-\frac{v}{u}\right)^{N-1}.\end{equation}
Maximizing this in terms of $v$ yields
\begin{equation}
\label{yuu}
v_{\mathrm{opt}}=\frac{u}{N}.
\end{equation}
Computation of $v_{\mathrm{opt}}$ requires that all transmitters
know the number of active users $N$ in the network. As far as all
channel gains are realizations of independent and continuous random
variables, the number of power levels in the PDF of noise plus
interference on any frequency sub-band at the receiver side of any
user is almost surely equal to $2^{N-1}$. Therefore, any receiver
can identify $N$ and send it to its corresponding transmitter
through a feedback link.

Setting $v=v_{\mathrm{opt}}$, the highest sum multiplexing gain
of the fair  FH scheme is given by
 \begin{equation}
 \label{bxxxx}
\sup_{v}\, \, \mathsf{SMG}(v,N)=\frac{1}{2}
u\left(1-\frac{1}{N}\right)^{N-1}.\end{equation} It is remarkable
that $\frac{u}{N}$ may not be a positive
integer. If we do not adopt the generalized FH scheme, then all
users must hop randomly over  sets of $\tilde{v} =
\max\left\{\lfloor\frac{u}{N}\rfloor,1 \right\}$ frequency
sub-bands. This results in a sum multiplexing gain of
$\frac{N}{2} \tilde{v} \left(1-\frac{\tilde{v}}{u}\right)^{N-1}$.
This is, in general, less than
$\frac{1}{2}u\left(1-\frac{1}{N}\right)^{N-1}$. By adopting the
generalized FH scheme in case $\frac{u}{N}\notin\mathbb{Z}$, each
user only needs to hop randomly over different sets of frequency
sub-bands of cardinality $\lfloor\frac{u}{N}\rfloor$ or
$\lceil\frac{u}{N}\rceil$. In fact, each user has two random
generators. The first random generator selects one of the numbers
$\lfloor\frac{u}{N}\rfloor$ and $\lceil\frac{u}{N}\rceil$ with
probabilities $\mu$ and $1-\mu$, respectively, such that
$\mu\lfloor\frac{u}{N}\rfloor+(1-\mu)\lceil\frac{u}{N}\rceil=\frac{u}{N}$
or equivalently $\mu=\lceil\frac{u}{N}\rceil-\frac{u}{N}$. Let us
assume the first random generator has selected a number
$a\in\{\lfloor\frac{u}{N}\rfloor,\lceil\frac{u}{N}\rceil\}$. Then,
the second random generator selects a subset of cardinality
$a$ among the $u$ frequency sub-bands. Doing this independently from
transmission slot to transmission slot, the sum multiplexing gain given
in (\ref{bxxxx}) is achieved.

\textit{Observation 1}- One might suggest another well-known utility
function that is popular in the game theory context, namely the
\textit{proportional fair function}, which is defined as $\sum_{i=1}^{N}\log
\mathscr{R}_{i}$ . We have
\begin{eqnarray}
\sum_{i=1}^{n}\log \mathscr{R}_{i} &\sim& \sum_{i=1}^{N}\log\left( \frac{1}{2} \bar{v}_{i}\prod_{\substack{k=1\\k\neq i}}^{N} \left(1-\frac{\bar{v}_{k}}{u} \right)\log \gamma\right) \notag\\
&=& \sum_{i=1}^{N}\log\left(\frac{1}{2} \bar{v}_{i}\prod_{\substack{k=1\\k\neq i}}^{N}\left(1-\frac{\bar{v}_{k}}{u}\right)\right)+N\log\log\gamma .
\end{eqnarray}
It can be easily verified that $\sum_{i=1}^{N}\log\left(\frac{1}{2}
\bar{v}_{i}\prod_{\substack{k=1\\k\neq
i}}^{N}\left(1-\frac{\bar{v}_{k}}{u}\right)\right)$ has an absolute
maximum at $\bar{v}_{i}=\frac{u}{N}$ for $1\leq i\leq N$.

\textit{Observation 2}- As we will discuss in more detail in the
next section, the number of active users in the system is in general
a random variable. Although users in the FH system can use their
knowledge about the number of active users to
adjust the hopping parameter (as explained earlier in this section),  one may devise a sub-optimal rule to
fix $v=v^{*}$ given by
\begin{equation}
\label{pof}
v^{*}=\arg\,\,\, \max_{v\in[0,u]}\mathrm{E}\left\{\mathsf{SMG}(v,N)\right\},
\end{equation}
where the expectation is with respect to the number of active users
in the network. This selection of $v$ by all users makes the system robust
against changes in the number of active users in the
network\footnote{In fact, the transmitters use their knowledge about
the instantaneous number of active users only to regulate their
transmission rate. This is explained in more details in
(\ref{hoolooo}).}. We call this version of the FH system the robust Frequency Hopping. In fact, in the robust FH scenario, there exists a \emph{global hopping parameter} $v$ where all users hop over a number $v$ of the $u$ frequency sub-bands. We remark that the rule in (\ref{pof}) is a particular design approach for the robust FH system.  In the next section, we consider another design rule based on maximizing the average of the minimum multiplexing gain per user in term of the number of active users in the network. 


  \section{Comparison of the robust FH scenario with the FD scheme}

In a centralized setup, under the condition that no user is aware of
the other users' codebooks and the number of users is fixed and
known to the central controller, it is shown in \cite{1} that if the crossover channel gains are sufficiently larger than the forward channel gains, then every Pareto optimal rate vector is realized by Frequency Division for all ranges of SNR. However,
in realistic scenarios, the number of active users is not fixed.
This degrades the performance of the FD scheme as it is designed for
a specified number of users. In particular, if the number of active users
is less than the designed target of the FD scheme, a considerable portion
of the spectrum may remain unused. This encourages us to compare the performance
of the proposed robust FH scheme with that of the FD scheme in a setup where
the number of active users is a random variable with a given distribution\footnote{Note that, as explicitly mentioned in the system model, the number of users is assumed to be
fixed for the whole transmission period of interest.}.

To perform the comparison, we introduce four different performance
measures. In the following definitions, the $\sup$ operation is over
possible adjustable parameters in the system, e.g., the hopping
parameter in the FH scenario.  All expectations are taken with respect to
$N$. We define $q_{n}\triangleq\Pr\{N=n\}$ for all $n\geq 0$.  It is
assumed that the maximum number of active users in the network is
$n_{\max}$, i.e.,  $\Pr\{N>n_{\max}\}=0$. We usually take
$q_{0}=0$ unless otherwise stated.
\begin{itemize}
 \item \textit{Average sum multiplexing gain}, which is defined as
\begin{eqnarray}
 \eta^{(1)} \triangleq \sup\lim_{\gamma \to \infty} \frac{\mathrm{E} \left\{\sum_{i=1}^N \mathscr{R}_i \right\}}{\log \gamma} = \sup\mathrm{E} \left\{\mathsf{SMG}\right\},
\end{eqnarray}
 where $\mathsf{SMG}=\lim_{\gamma\to\infty}\frac{\sum_{i=1}^{N}\mathscr{R}_{i}}{\log\gamma}$ is the sum multiplexing gain.
\item \textit{Average minimum multiplexing gain per user}, which is defined as
\begin{eqnarray}
 \eta^{(2)} \triangleq\sup \lim_{\gamma \to \infty} \frac{\mathrm{E} \left\{\min_{1\leq i\leq N}\mathscr{R}_{i} \right\}}{\log \gamma}.
\end{eqnarray}
\item \textit{Minimum nonzero multiplexing gain per user}, which is defined as
\begin{eqnarray}
 \eta^{(3)} \triangleq\min_{n: q_n \neq 0} \quad \min_{\substack {N_{\mathrm{serv}}=n\\1 \leq i \leq n}} \quad \lim_{\gamma \to \infty}  \frac{\mathscr{R}_i}{\log \gamma}
\end{eqnarray}
 where $N_{\rm{serv}}$ denotes the number of active users receiving service (i.e., their multiplexing gain is strictly positive).\item \textit{Service capability}, which is defined as
\begin{eqnarray} \label{eta4}
 \eta^{(4)} \triangleq \sup\mathrm{E} \left\{ \frac{N_{\rm{serv}}}{N}\right\}.
\end{eqnarray}

\end{itemize}

The FD system is designed to service, at most, a certain number of
active users. We denote this design target in the FD scheme by
$n_{\mathrm{des}}$. Therefore, the spectrum is divided to
$n_{\mathrm{des}}$ \emph{bands} where each band contains
$\frac{u}{n_{\mathrm{des}}}$ frequency sub-bands. This requires that
$u$ is divisible by $n_{\mathrm{des}}$, which is assumed to be the
case to guarantee fairness. Each user that becomes active occupies
an empty band. If there is no empty band, no service is available.
In case $n_{\max}$ is finite, the central controller in the FD
system sets $n_{\mathrm{des}}=n_{\max}$ to ensure that all users can
receive service upon activation. In case $n_{\max}$ is not a finite
number, the central controller sets $n_{\mathrm{des}}=u$ to guarantee
that as many users receive service as possible. Therefore, $n_{\mathrm{des}}=\min\{n_{\max},u\}$. In fact, we will show that selecting $n_{\mathrm{des}}=\min\{n_{\max},u\}$ maximizes the service capability in the FD system.

We remark that due to the nature of the robust FH scheme, as far as users  hop over a proper subset of size $v$ of the $u$ sub-bands, all
users receive service, while if $v=u$ and $N>1$, no user receives
service, i.e.,  the multiplexing gain achieved by any active user is
zero. As such, to get the largest service capability in the FH scenario, we require $v\in(0,u)$. As an example, if $v^{*}$ in (\ref{pof}) is equal to $u$, the service capability will be less than $1$. To avoid this, we set the global hopping parameter $v=v^{*}-\varepsilon=u-\varepsilon$ for sufficiently small $\varepsilon$ such that the performance of the robust FH is still above the performance of the FD scenario.   

$\bullet$ \textit{ Average sum multiplexing gain}

This measure
is a meaningful tool of comparison if $n_{\max}<\infty$.
Hence, we assume $n_{\max}$ is a finite number and $u$ is a multiple
of $n_{\max}$ in this subsection. It is easily seen that the
sum multiplexing gain in the FD scenario is
 \begin{equation}
 \label{pooohj}
 \mathsf{SMG}_{\mathrm{FD}}(n_{\mathrm{des}},N)=\left\{\begin{array}{cc }
    \frac{N}{2} \frac{u}{n_{\mathrm{des}}}  & N\leq n_{\mathrm{des}}   \\
       \frac{u}{2}  &   N>n_{\mathrm{des}}
\end{array}\right..
 \end{equation}
 Noting (\ref{yu}),  $\mathsf{SMG}_{\mathrm{FH}}(v,N)$
is given by
 \begin{equation}
 \label{llll}
 \mathsf{SMG}_{\mathrm{FH}}(v,N)=\frac{1}{2}Nv\left(1-\frac{v}{u}\right)^{N-1}. \end{equation}
Since the number of active users $N$ is a global knowledge, all
users can choose $v=v_{\mathrm{opt}}=\frac{u}{N}$ to achieve the
maximum sum multiplexing gain.  However, as mentioned earlier, a robust
hopping strategy against changes in the number of active users is
the one given in (\ref{pof}).
It is notable that although the value of $v$ is fixed at $v^{*}$,
all users regulate their rates based on the instantaneous number of
active users to avoid transmission failure. Using the lower bound on
the achievable rate of the $i^{th}$ user given in (\ref{uuku}), the
$i^{th}$ user selects its actual rate $R_{i}$ as
\begin{equation}
\label{hoolooo}
R_{i}=\frac{v^{*}}{2}\log\left(\frac{\left(\frac{v^{*}}{u}\right)^{-\frac{2(N-1)v^{*}}{u}}\left(1-\frac{v^{*}}{u}\right)^{-2(N-1)\left(1-\frac{v^{*}}{u}\right)}\vert
h_{i,i}\vert^{2}\gamma}{v^{*}\left(1+\frac{\sum_{j\neq
i}|h_{j,i}|^{2}\gamma}{v^{*}}\right)^{1-\left(1-\frac{v^{*}}{u}\right)^{N-1}}}+1\right).
\end{equation}
It is seen that the quantities the $i^{th}$ transmitter needs to
evaluate $R_{i}$ are $|h_{i,i}|$, $\sum_{j\neq i}|h_{j,i}|^{2}$ and
$N$. The $i^{th}$ receiver sends these required data to the
transmitter via a feedback link.

We present an example to compare the performance of FH with that of FD in terms of
$\eta^{(1)}$.

 \textit{Example 1}-
Let us consider a network where $n_{\max}=2$. The central controller
in the FD system sets $n_{\mathrm{des}}=2$, and according to
(\ref{pooohj}),
$ \eta_{\mathrm{FD}}^{(1)}=\mathrm{E}\{\mathsf{SMG}_{\mathrm{FD}}(2,N)\}=q_{1}\frac{u}{4}+q_{2}\frac{u}{2}=\frac{q_{1}+2q_{2}}{4}u$. Based on (\ref{llll}),
 $\mathrm{E}\{\mathsf{SMG}_{\mathrm{FH}}(v,N)\}=\frac{1}{2}q_{1}v+q_2v\left(1-\frac{v}{u}\right)$.  
Using this in (\ref{pof}),
\begin{equation}
v^{*}=\arg\max_{v\in[0,u]}\mathrm{E}\{\mathsf{SMG}_{\mathrm{FH}}(v,N)\}=\left\{\begin{array}{cc}
   \frac{q_{1}+2q_{2}}{4q_{2}}u   & q_{1}\leq 2q_{2}   \\
    u  &  q_{1}>2q_{2}
\end{array}\right..
\end{equation}
Therefore,
\begin{eqnarray}
\eta_{\mathrm{FH}}^{(1)}=\sup_{v\in[0,u]}\mathrm{E}\{\mathsf{SMG}_{\mathrm{FH}}(v,N)\}=\mathrm{E}\{\mathsf{SMG}_{\mathrm{FH}}(v^{*},N)\}
=\left\{\begin{array}{cc}
    \frac{(q_{1}+2q_{2})^{2}}{16q_{2}}u  &  q_{1}\leq 2q_{2}  \\
    \frac{q_{1}}{2}u  & q_{1}>2q_{2}
\end{array}\right..
\end{eqnarray}
It is easy to see that
$\eta_{\mathrm{FH}}^{(1)}>\eta_{\mathrm{FD}}^{(1)}$ if and only if
$q_{1}>2q_{2}$, or equivalently, $q_{1}>\frac{2}{3}$. We note that
in this case $v^{*}=u$, i.e., all users spread their power on the
whole spectrum and no hopping is performed. This makes service capability be strictly less than $1$ because, if both users are active, non of them receive service. As such, we take $v=u-\varepsilon$. To ensure that the performance of the robust FH scenario is above that of the FD system, we require
\begin{equation}
\label{lki}
\frac{1}{2}q_{1}(u-\varepsilon)+q_2(u-\varepsilon)\left(1-\frac{u-\varepsilon}{u}\right)>\frac{q_{1}+2q_{2}}{4}u.
\end{equation}
As far as $\varepsilon<\frac{u}{2}$, (\ref{lki}) is equivalent to
$q_{1}>2q_{2}\frac{1-\frac{2\varepsilon}{u}\left(1-\frac{\varepsilon}{u}\right)}{1-\frac{2\varepsilon}{u}}$. This is a more restrictive condition than $q_{1} > 2q_{2}$ which is the cost paid for having full service capability. However, for $\varepsilon \ll u$ the two regions of $(q_{1},q_{2})$ are almost the same.
 $\square$

In \cite{kami-1}, it is shown that in case $n_{\max}=3$,
there exists a probability set of $(q_{1},q_{2},q_{3})$  on the number of active users that makes  FH achieve a
higher performance compared to FD in terms of $\eta^{(1)}$ while $v^{*}$ is
strictly less than $u$.

$\bullet$ \textit{Average minimum multiplexing gain per user}

This measure  can also be written as
  \begin{equation}
  \eta^{(2)}=\sup\mathrm{E} \left \lbrace \frac{\mathsf{SMG}}{N}\mathbb{1}(N_{\mathrm{serv}}=N)\right \rbrace.
  \end{equation}
In fact, if $N_{\mathrm{serv}} \neq N$, there exists at least one
user that achieves no multiplexing gain. Therefore, the minimum
multiplexing gain per user is zero in this case. However, if
$N_{\mathrm{serv}}=N$, all users achieve a nonzero multiplexing
gain. This measure can be used whether $n_{\max}$ is finite or
infinite.

In case of the FH scenario, the rule to choose the optimum value of the global hopping parameter
$v$, denoted by $v^{\dagger}$, is given by
 \begin{equation}
 \label{pof2}
v^{\dagger}=\arg\,\,\, \max_{v\in[0,u]}\mathrm{E}\left \lbrace \frac{\mathsf{SMG}_{\mathrm{FH}}(v,N)}{N}\mathbb{1}(N=N_{\mathrm{serv}})\right \rbrace.
\end{equation}
In this case, the actual transmission rate of the $i^{th}$ user is given by (\ref{hoolooo}) where $v^{*}$ is replaced by $v^{\dagger}$.

\textit{Example 2}- Considering the same setup in example 1, as
$n_{\max}<\infty$, we have $N_{\mathrm{serv,FD}}=N$. Hence, we have
$\eta_{\mathrm{FD}}^{(2)}=\frac{1}{2}\frac{u}{2}q_1 +
\frac{1}{2}\frac{u}{2} q_{2}=\frac{u}{4}$. In case of the FH scheme,
\begin{equation}
\label{ }
\mathbb{1}(N_{\mathrm{serv,FH}}=N)=\left\{\begin{array}{cc}
    1  &  \textrm{$N=1$ or ($N>1$ and $v\neq u$) }  \\
    0  &   \mathrm{oth.}
\end{array}\right..
\end{equation}

Hence,
\begin{eqnarray}
\mathrm{E}\left\{\frac{\mathsf{SMG}_{\mathrm{FH}}(v,N)}{N}\mathbb{1}(N_{\mathrm{serv}}=N)\right\}&=& \mathrm{E}\left\{\frac{\mathsf{SMG}_{\mathrm{FH}}(v,N)}{N}\mathbb{1}(N_{\mathrm{serv}}=N)\Big| N=1\right\}\Pr\{N=1\}\notag\\&&+\mathrm{E}\left\{\frac{\mathsf{SMG}_{\mathrm{FH}}(v,N)}{N}\mathbb{1}(N_{\mathrm{serv}}=N) \Big|N=2\right\}\Pr\{N=2\}\notag \\
&=& \frac{1}{2}q_{1}v+\frac{1}{2}q_{2}v\left(1-\frac{v}{u}\right)\mathbb{1}(v\neq u)\notag\\
&=& \frac{1}{2}q_{1}v+\frac{1}{2}q_{2}v\left(1-\frac{v}{u}\right).
\end{eqnarray}
Hence,
\begin{equation}
\label{ }
v^{\dagger}=\arg\max_{v\in(0,u]}\left\{q_{1}v+q_{2}v\left(1-\frac{v}{u}\right)\right\},\end{equation}
which yields
\begin{equation}
\label{ }
v^{\dagger}=\left\{\begin{array}{cc}
    \frac{u}{2q_{2}}  &  2q_{2}>1  \\
    u  &   2q_{2}\leq 1
\end{array}\right..
\end{equation}

As such,
\begin{eqnarray}
\eta_{\mathrm{FH}}^{(2)}=\left\{\begin{array}{cc}
   \frac{1}{8q_{2}}u &  2q_{2}>1  \\
    \frac{1}{2}q_{1}u  & 2q_{2}\leq 1
\end{array}\right..
\end{eqnarray}
It is easy to see that $\eta_{\mathrm{FH}}^{(2)}>\eta_{\mathrm{FD}}^{(2)}$
if and only if $2q_{2}<1$, or equivalently $q_{1}>\frac{1}{2}$. However, in this
case $v^{\dagger}=u$. Hence, to make the service capability be $1$, we choose the global hopping parameter $v=u-\varepsilon$. To ensure that FH still outperforms FD in terms of the average minimum multiplexing gain per user, we require, 
\begin{equation}
\label{ }
 \frac{1}{2}q_{1}(u-\varepsilon)+\frac{1}{2}q_{2}(u-\varepsilon)\left(1-\frac{u-\varepsilon}{u}\right)>\frac{u}{4}.
 \end{equation}
 This is equivalent to $2q_{2}<\frac{2}{1-\frac{\varepsilon}{u}}\left(1-\frac{1}{2\left(1-\frac{\varepsilon}{u}\right)}\right)$. 
 $\square$
 
In the next example, we provide a case where
$\eta_{\mathrm{FH}}^{(2)}>\eta_{\mathrm{FD}}^{(2)}$
while $v^{\dagger}$ is strictly less than $u$.

\textit{Example 3}- Let $n_{\max}<\infty$. In this example, we aim to derive a
sufficient condition on $\{q_{n}\}_{n=1}^{n_{\max}}$ such that
$\eta_{\mathrm{FH}}^{(1)}> \eta_{\mathrm{FD}}^{(1)}$ or $\eta_{\mathrm{FH}}^{(2)}> \eta_{\mathrm{FD}}^{(2)}$.

\textit{Case 1-} Let us consider the measure $\eta^{(1)}$.  We have the following result.
\begin{proposition}
As far as
\begin{equation}\label{goolj}\mathrm{E}\{N\}< \frac{1}{2}\ln\left((e^{2}-1)n_{\max}\right),\end{equation} we have $\eta_{\mathrm{FD}}^{(1)}<\eta_{\mathrm{FH}}^{(1)}$.
\end{proposition}
\begin{proof}
See Appendix B.
\end{proof}
For example, if $n_{\max}=2$, (\ref{goolj}) gives $\mathrm{E}\{N\} \leq 1.274$, or equivalently $q_{1} \geq 0.726$.
By example 1, we notice that $\eta_{\mathrm{FH}}^{(1)}\geq \eta_{\mathrm{FD}}^{(1)}$ if and only if
$q_{1} \geq 0.667$.

\textit{Case 2-} As for $\eta^{(2)}$, along the same lines leading
to (\ref{goolj}),  a sufficient condition for
$\eta_{\mathrm{FH}}^{(2)}>\eta_{\mathrm{FD}}^{(2)}$ is given in the
following Proposition.
\begin{proposition}
As far as
\begin{equation}
\label{vegggg}
\frac{1}{\mathrm{E}\{N\}}\left(1-\frac{1}{\mathrm{E}\{N\}}\right)^{\mathrm{E}\{N\}-1}>\frac{1}{n_{\max}},\end{equation}
we have $\eta_{\mathrm{FD}}^{(2)}<\eta_{\mathrm{FH}}^{(2)}$.\end{proposition}
\begin{proof}
See Appendix C.
\end{proof}


For example, if $n_{\max}=10$, $q_{1}=0.22$,
$q_{2}=q_{3}=q_{4}=0.24$ and $q_{5}=q_{6}=\cdots=q_{10}=0.01$, one
has $\mathrm{E}\{N\}=2.78$, which satisfies (\ref{vegggg}).
Therefore, we conclude
$\eta_{\mathrm{FH}}^{(2)}>\eta_{\mathrm{FD}}^{(2)}$. Computing these
quantities directly, we get $\eta_{\mathrm{FD}}^{(2)}=\frac{u}{16}$
and
\begin{eqnarray}
\label{ }
\eta_{\mathrm{FH}}^{(2)}&=&\frac{1}{2}\max_{v\in[0,u]}\left\{v\sum_{n=1}^{10}q_{n}\left(1-\frac{v}{u}\right)^{n-1}\right\}\notag\\
&\stackrel{(a)}{=}&\frac{1}{2}u\max_{\omega_{v}\in[0,1]}(1-\omega_{v})\left(0.22+0.24(\omega_{v}+\omega_{v}^{2}+\omega_{v}^{3})+0.01(\omega_{v}^{4}+\omega_{v}^{5}+\omega_{v}^{6}+\omega_{v}^{7}+\omega_{v}^{8}+\omega_{v}^{9})\right)\notag\\
&\stackrel{(b)}{=}&0.1121u
\end{eqnarray}
where in $(a)$, we define $\omega_{v}\triangleq 1-\frac{v}{u}$ and $(b)$ is
obtained by setting $\omega_{v}=0.28$, or equivalently
$v=v^{\dagger}=0.72u$. This yields
$\frac{\eta_{\mathrm{FH}}^{(2)}}{\eta_{\mathrm{FD}}^{(2)}}=1.7936$.

\textit{Example 4}- In this example, we assume a Poisson
distribution on the number of active users, i.e.,
$q_{n}=\frac{e^{-\lambda}\lambda^{n}}{n!}$, $n\geq 0$. This
assumption corresponds to the scenario where potentially a large
number $n_{\max}$ of users may share the spectrum. However, the
activation probability $p$ of each user is very small. One can well
approximate the number of active users in the network by a Poisson
random variable with parameter $\lambda = pn_{\max}$. We have
\begin{eqnarray}
\label{ }
\mathrm{E}\left\{\frac{\mathsf{SMG}_{\mathrm{FH}}(v,N)}{N}\mathbb{1}(N_{\mathrm{serv,FH}}=N)\right\}&\stackrel{(a)}{=}&\mathrm{E}\left\{\frac{\mathsf{SMG}_{\mathrm{FH}}(v,N)}{N}\right\}\notag\\&=&\frac{1}{2}\sum_{n=1}^{\infty}\frac{e^{-\lambda}\lambda^{n}}{n!}\left(v\left(1-\frac{v}{u}\right)^{n-1}\right)\notag\\
&=&\frac{1}{2}\frac{v}{1-\frac{v}{u}}\sum_{n=1}^{\infty}\frac{e^{-\lambda}\lambda^{n}}{n!}\left(1-\frac{v}{u}\right)^{n}\notag\\
&\stackrel{(b)}{=}&\frac{1}{2}\frac{v}{1-\frac{v}{u}}\left(e^{\lambda\left(e^{\ln\left(1-\frac{v}{u}\right)}-1\right)}-e^{-\lambda}\right) \notag\\
&=& \frac{e^{-\lambda}(1-\omega_{v})\left(e^{\lambda\omega_{v}}-1\right)}{2\omega_{v}}u.
\end{eqnarray}
In the above equation, $(a)$ results from the fact that $\mathbb{1}(N_{\mathrm{serv,FH}}=N)=0$ whenever $v=u$ and $N >1$, however, $\mathsf{SMG}_{\mathrm{FH}}(v,N) =0$ in this case. $(b)$ follows by the fact that
$E\{e^{tN}\}=e^{\lambda(e^{t}-1)}$ for any $t$ and
$\omega_{v}=1-\frac{v}{u}$ as defined in example 3. It can be easily
seen that the optimal $\omega_{v}$ satisfies the nonlinear equation
$e^{-\lambda\omega_{v}}=1-\lambda\omega_{v}+\lambda\omega_{v}^{2}$.
Solving this for $\omega_{v}$, we find out that $v^{\dagger}$ is not
equal to $u$ for all $\lambda >2$. The following table lists the optimum values of
$\omega_{v}$, i.e., $\omega_{v^{\dagger}}$, the values of
$v^{\dagger}$ and also the corresponding average minimum
multiplexing gain per user $\eta_{\mathrm{FH}}^{(2)}$ for
$\lambda\in\{3,\cdots,10\}$.
\begin{equation}\begin{tabular}{|c|c|c|c|c|c|c|c|c|c|c|}
\hline
$\lambda$&$3$& $4$& $5$ &$6$&$7$&$8$&$9$&$10$\\
\hline
$\omega_{v^{\dagger}}$& $0.4536$&$0.6392$&$0.7347$&$0.7912$&$0.828$&$0.8537$&$0.8727$&$0.8873$\\
\hline
$v^{\dagger}$&$0.5464u$& $ 0.3608u$&   $ 0.2653u$&$ 0.2088u$&   $0.1720u$&  $  0.1463u$&$0.1273u$&$0.1127u$\\
\hline
$\eta_{\mathrm{FH}}^{(2)}$& $0.0869u$&$0.0615u$&$0.0467u$&$0.0374u$&$0.0311u$&$0.0266u$&$0.0232u$&$0.0206u$\\
\hline
\end{tabular}.\end{equation}

In order to provide fairness among the users, the FD system tries to
serve as many users as it can. Since it is not possible to serve
more than $u$ users and $n_{\max} \gg u$, the central controller
sets $n_{\mathrm{des}}=u$. Therefore, $N_{\mathrm{serv,FD}}<N$ if
and only if $N>u$. Using this and  by (\ref{pooohj}),
\begin{eqnarray}
\eta_{\mathrm{FD}}^{(2)}&=&\mathrm{E}\left\{\frac{\mathsf{SMG}_{\mathrm{FD}}(n_{\mathrm{des}},N)}{N}\mathbb{1}(N_{\mathrm{serv,FD}}=N)\right\}\notag\\
&=&\mathrm{E}\left\{\frac{\mathsf{SMG}_{\mathrm{FD}}(u,N)}{N}\Bigg|N\leq u\right\}\Pr\{N\leq u\}\notag\\&=&\frac{1}{2}\sum_{n=1}^{u}\frac{e^{-\lambda}\lambda^{n}}{n!}.
\end{eqnarray}
We have sketched $\eta_{\mathrm{FH}}^{(2)}$ and $\eta_{\mathrm{FD}}^{(2)}$
in terms of $\lambda$ in fig. 1 and fig. 2 for the cases $u=7$ and $u=20$, respectively.
It is noticeable that $\eta_{\mathrm{FH}}^{(2)}$ scales linearly with $u$.
However, $\eta_{\mathrm{FD}}^{(2)}$ is always less than $\frac{1}{2}$ no matter how
large $u$ is. Thus, as $u$ increases, the advantage of FH over FD becomes more apparent.
$\rightmark{\square}$
\begin{figure}[h!b!t]
  \centering
  \includegraphics[scale=.6] {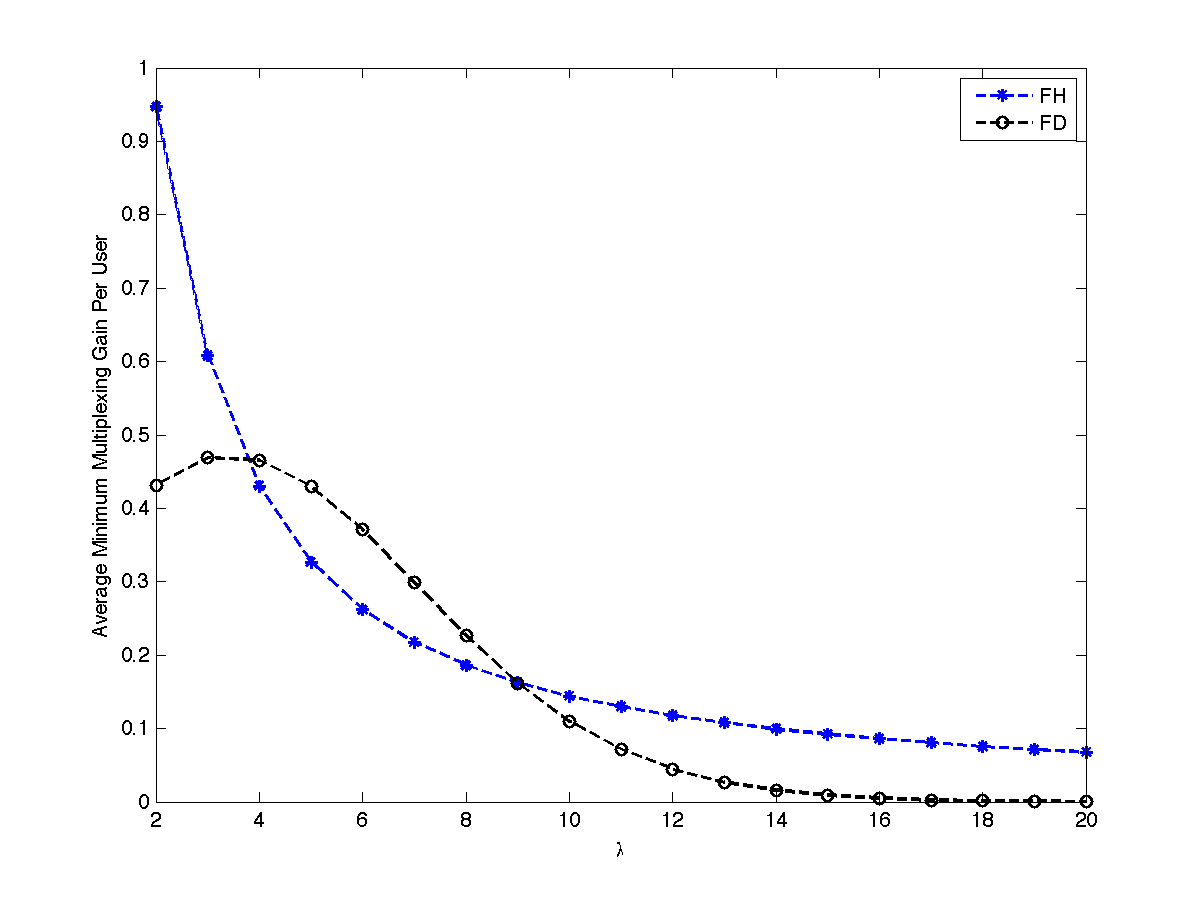}
  \caption{Curves of $\eta_{\mathrm{FH}}^{(2)}$ and $\eta_{\mathrm{FD}}^{(2)}$ in terms of $\lambda$ in a network with $u=7$ sub-bands.}
  \label{figapp1}
 \end{figure}
 \begin{figure}[h!b!t]
  \centering
  \includegraphics[scale=.6] {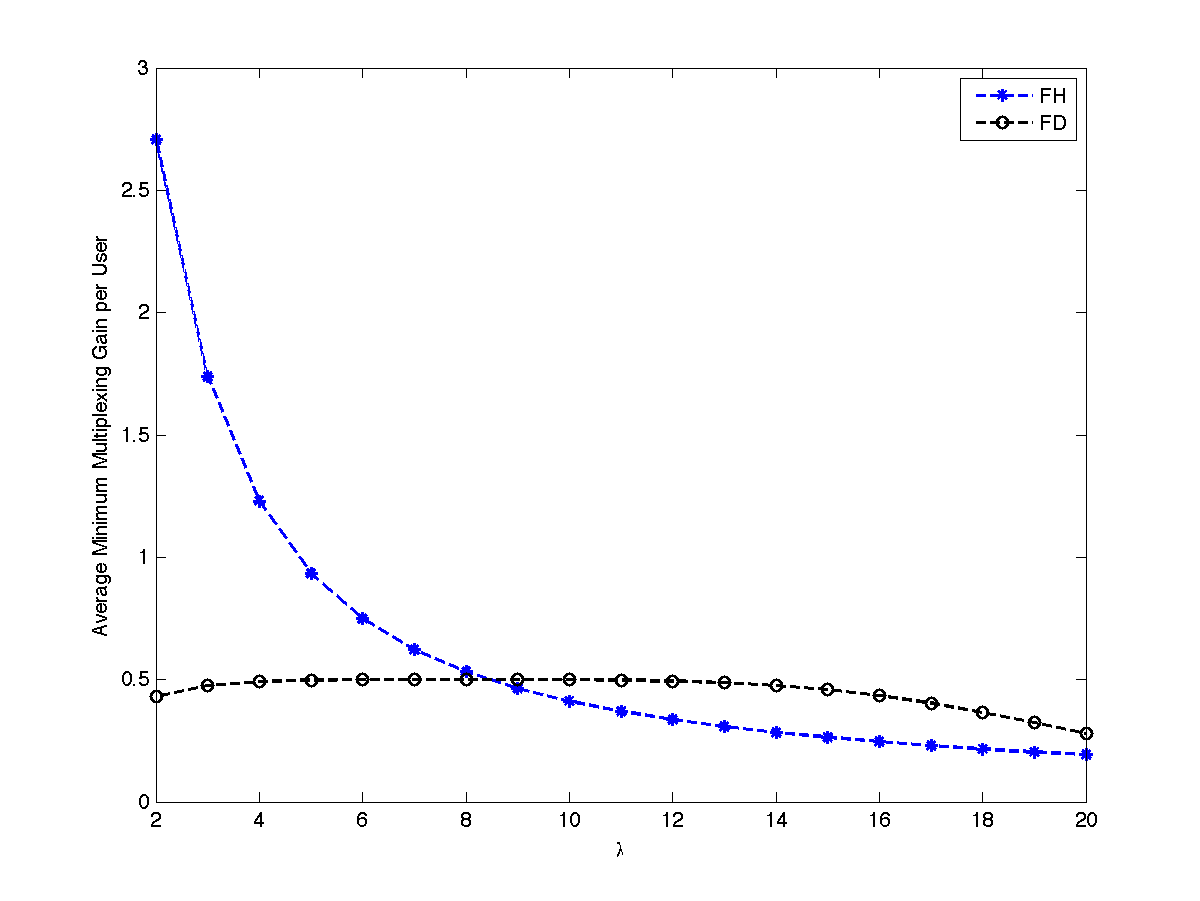}
  \caption{Curves of $\eta_{\mathrm{FH}}^{(2)}$ and $\eta_{\mathrm{FD}}^{(2)}$ in terms of $\lambda$ in a network with $u=20$ sub-bands.}
  \label{figapp1}
 \end{figure}
\newpage

$\bullet$ \textit{Minimum nonzero multiplexing gain per user}

 The minimum nonzero multiplexing gain per user is the smallest nonzero
multiplexing gain that a user in the network attains for different
realizations in terms of the number of active users. Assuming
$n_{\max}<\infty$, this happens when there are exactly $n_{\max}$
active users in the system. As the FD system is already designed to
handle the case where $n_{\max}$ users are present in the network,
the minimum multiplexing gain per user is automatically higher in FD
as compared to FH.  Setting $n_{\mathrm{des}}=n_{\max}$, we have
$\eta_{\mathrm{FD}}^{(3)}=\frac{{\mathsf{SMG}}_{\mathrm{FD}}(u,n_{\max})}{n_{\max}}=\frac{u}{2
n_{\max}}$. In the case of FH, we assume that all users select
$v=\frac{u}{n_{\max}}$. Hence,
$\eta_{\mathrm{FH}}^{(3)}=\frac{{\mathsf{SMG}}_{\mathrm{FH}}\left(\frac{u}{n_{\max}},n_{\max}\right)}{n_{\max}}=\frac{u}{2
n_{\max}}\left(1-\frac{1}{n_{\max}}\right)^{n_{\max}-1}$. Clearly,
$\frac{1}{e}\leq
\frac{\eta_{\mathrm{FH}}^{(3)}}{\eta_{\mathrm{FD}}^{(3)}}\leq 1$ as
$\left(1-\frac{1}{n_{\max}}\right)^{n_{\max}-1}$ approaches
$\frac{1}{e}$ from above by increasing $n_{\max}$. Therefore, the
loss incurred in terms of $\eta^{(3)}$ for the FH system is always less
than $\frac{1}{e}$.

 $\bullet$ \textit{Service capability}

Service capability demonstrates the fraction of users receiving
service among the whole active users in the network. As mentioned earlier, a user is said
to receive service whenever the achieved multiplexing gain of the
user is nonzero. In the FD scenario, if $N>u$, then a fraction of
users cannot share the spectrum. However, in case $\Pr\{N\leq
u\}=1$, the FD scheme achieves the full service capability. As for
the FH scheme, we already know that as far as all users hop over
proper subsets of the sub-bands, every user achieves a nonzero
multiplexing gain. The following examples offer comparisons between
FD and FH in terms of the service capability.

\textit{Example 5}- In this example, we consider a setup where
$n_{\max}<\infty$. The central controller in FD simply sets
$n_{\mathrm{des}}=n_{\max}$ and the service capability is always
equal to $1$. The number of served users $N_{\mathrm{serv,FH}}$ in
the FH scenario can be written as
\begin{equation}
\label{doogh}
N_{\mathrm{serv,FH}}=\left\{\begin{array}{cc}
   N   &  \textrm{$N=1$ or ($N>1$ and $v\neq u$)}  \\
    0  &  \textrm{oth.}
\end{array}\right..
\end{equation}
Therefore, as far as $v\neq u$, we have $N_{\mathrm{serv,FH}}=N$ and
the service capability is one. This shows that to achieve the
maximum service capability in a system where $n_{\max}>1$, the
hopping parameter $v$ must be strictly less than $u$.
$\rightmark{\square}$

\textit{Example 6}- In this example, we provide a case where
$n_{\max}$ is not finite. Let us assume the distribution of the
number of active users in the network is a Poisson distribution with
parameter $\lambda$, i.e., $q_{n}=\frac{\lambda^n e^{-\lambda}}{n!}$
for $n\geq 0$ where $\lambda > 1$. Let us compute $v^{*}$ for the FH
scenario. We have,
\begin{eqnarray}
\mathrm{E}\{\mathsf{SMG}_{\mathrm{FH}}(v,N)\}&=&\frac{1}{2}\sum_{n=1}^{\infty}\frac{e^{-\lambda}\lambda^{n}}{n!}\left(nv\left(1-\frac{v}{u}\right)^{n-1}\right)\notag\\
&=&\frac{1}{2}v\sum_{n=1}^{\infty}\frac{e^{-\lambda}\lambda^{n}}{(n-1)!}\left(1-\frac{v}{u}\right)^{n-1}\notag\\
&=&\frac{1}{2}\lambda v\sum_{n=0}^{\infty}\frac{e^{-\lambda}\lambda^{n}}{n!}\left(1-\frac{v}{u}\right)^{n}\\
&=& \frac{1}{2}\lambda ve^{-\frac{\lambda v}{u}}. \label{bil2}
\end{eqnarray}
Thus,
\begin{equation}
\label{pooskooshi}
v^{*}=\arg\max_{v}\mathrm{E}\{\mathsf{SMG}_{\mathrm{FH}}(v,N)\}=\frac{u}{\lambda}.\end{equation}
Since $\lambda\neq 1$, we get $v^{*}\neq u$. Thus, choosing
$v=v^{*}$ maximizes
$\mathrm{E}\left\{\frac{N_{\mathrm{serv,FH}}}{N}\right\}$ and
$\mathrm{E}\{\mathsf{SMG}_{\mathrm{FH}}(v,N)\}$ simultaneously,
i.e., $\eta_{\mathrm{FH}}^{(4)}=1$.

In the FD system, $N_{\mathrm{serv,FD}}$ is given by
\begin{equation}
\label{pmnbv}
N_{\mathrm{serv,FD}}=\left\{\begin{array}{cc}
    N  & N\leq n_{\mathrm{des}}   \\
     n_{\mathrm{des}} &   N>n_{\mathrm{des}}
\end{array}\right..
\end{equation}
Thus,
\begin{eqnarray}
\mathrm{E}\left\{\frac{N_{\mathrm{serv,FD}}}{N}\right\}&=&\Pr\{ N \leq n_{\mathrm{des}}\} +n_{\mathrm{des}} \sum_{n=n_{\mathrm{des}}+1}^{\infty}\frac{q_{n}}{n}\notag\\
&=&1-\Pr\{N\geq n_{\mathrm{des}}+1\}+n_{\mathrm{des}} \sum_{n=n_{\mathrm{des}}+1}^{\infty}\frac{q_{n}}{n}\notag\\
&=&1-\sum_{n=n_{\mathrm{des}}+1}q_{n}\left(1-\frac{n_{\mathrm{des}}}{n}\right).
\end{eqnarray}
By this expression, it is clear that to maximize
$\mathrm{E}\left\{\frac{N_{\mathrm{serv,FD}}}{N}\right\}$, one must
select $n_{\mathrm{des}}$ as large as possible. This basically
justifies the assumption we made about selecting
$n_{\mathrm{des}}=u$ in the FD scheme in the case where $n_{\max}$
is not finite. Thus,
\begin{equation}
\label{ }
\eta_{\mathrm{FD}}^{(4)}=1-\sum_{n=u+1}q_{n}\left(1-\frac{u}{n}\right)\end{equation}
For instance, in the case of $u=5$ and $\lambda=3$, we have
$\eta_{\textrm{FD}}^{(4)}=0.9806$. $\rightmark{\square}$

\section{Adaptive Frequency Hopping}

The results of the previous section are obtained based on the
assumption that the hopping parameter $v$ is fixed and is not
adaptively changed based on the number of active users. The
performance of the FH system can be improved by letting the
transmitters adapt their hopping parameter based on the number of
active users using (\ref{yuu}). We refer to this scenario as
Adaptive Frequency Hopping (AFH). In the following example, we study
the performance improvement offered by AFH over FH in terms of
$\eta^{(1)}$ and $\eta^{(2)}$.

\textit{Example 7}- Let us assume that the number of active users is
a Poisson random variable with parameter  $\lambda>1$.  We already
have
\begin{eqnarray} \label{bil}
\eta_{\mathrm{FH}}^{(1)}=\frac{u}{2e},
\end{eqnarray}
while by (\ref{bxxxx}),
\begin{eqnarray}
\label{fag}
\eta_{\mathrm{AFH}}^{(1)}  =  \frac{u}{2}\sum_{n=1}^{\infty}\frac{e^{-\lambda} \lambda^{n}}{n!} \Big(1-\frac{1}{n}\Big)^{n-1}.
\end{eqnarray}
Figure \ref{fig6} shows the plots of $\eta_{\mathrm{FH}}^{(1)}$ and
$\eta_{\mathrm{AFH}}^{(1)}$ versus $\lambda$ for $u=10$. It is
observed that $\eta_{\mathrm{FH}}^{(1)}$ does not change with
$\lambda$, while $\eta_{\mathrm{AFH}}^{(1)}$ decreases by increasing
$\lambda$. This indicates that in a crowded network (large
$\lambda$), AFH does not provide any significant advantage over FH
in terms of $\eta^{(1)}$.
  \begin{figure}[h!b!t]
  \centering
  \includegraphics[scale=.6] {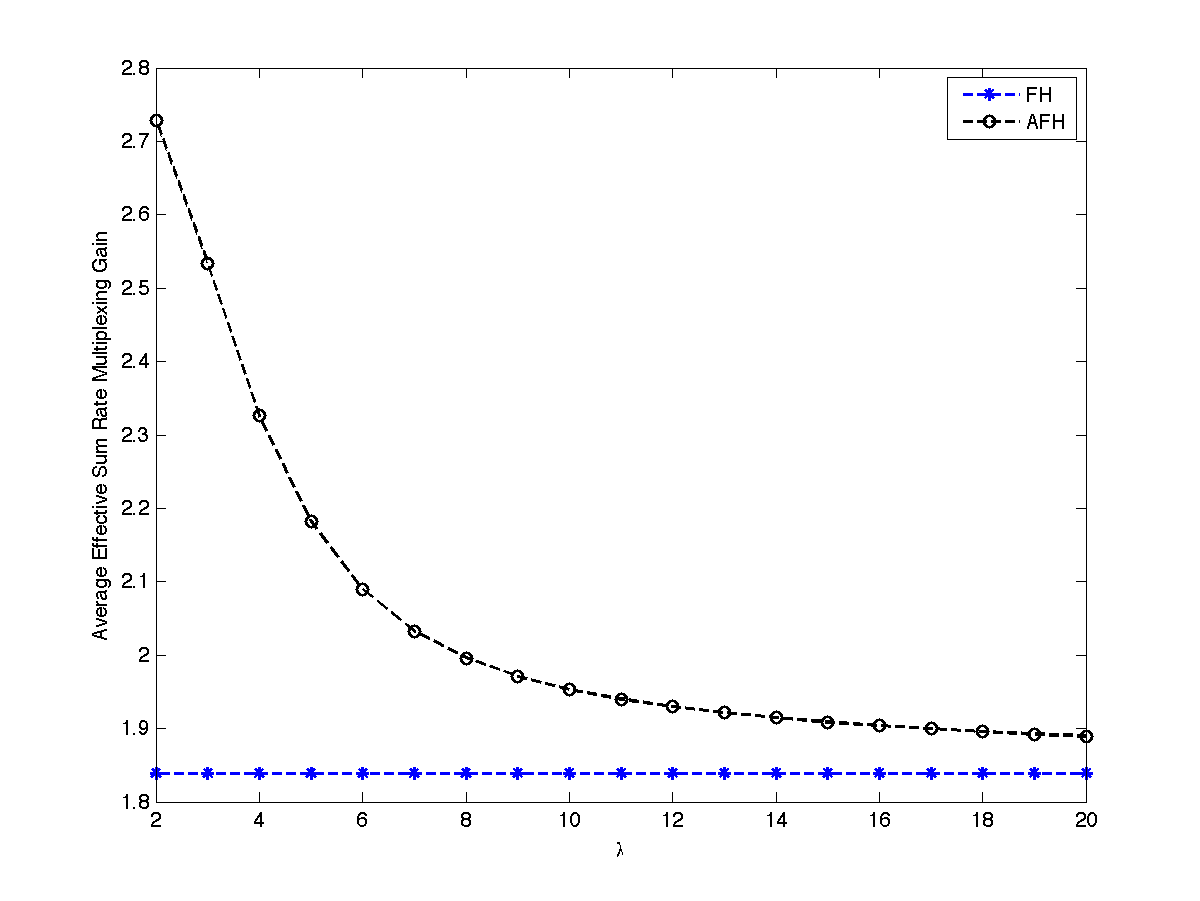}
  \caption{$\eta_{\mathrm{AFH}}^{(1)}$ versus $\eta_{\mathrm{FH}}^{(1)}$ for $u=10$. }
  \label{fig6}
 \end{figure}
 
We have already calculated $\eta_{\mathrm{FH}}^{(2)}$ in example 4
in a system where $3\leq\lambda\leq10$. However, in case of AFH,
 \begin{eqnarray}
\label{fag}
\eta_{\mathrm{AFH}}^{(2)}= \frac{u}{2}\sum_{n=1}^{\infty}\frac{e^{-\lambda} \lambda^{n}}{n!}\frac{1}{n}\Big(1-\frac{1}{n}\Big)^{n-1}.
\end{eqnarray}
Figure \ref{fig7} presents the plots of $\eta_{\mathrm{FH}}^{(2)}$
and $\eta_{\mathrm{AFH}}^{(2)}$ versus $\lambda$ for $u=10$. Both
$\eta_{\mathrm{FH}}^{(2)}$ and $\eta_{\mathrm{AFH}}^{(2)}$ decrease
by increasing $\lambda$. However, the ratio
$\frac{\eta_{\mathrm{AFH}}^{(2)}}{\eta_{\mathrm{FH}}^{(2)}}$
decreases as $\lambda$ increases. This indicates that for large
values of $\lambda$, AFH does also not provide any significant advantage
over FH in terms of $\eta^{(2)}$. $\rightmark{\square}$
\begin{figure}[h!b!t]
  \centering
  \includegraphics[scale=.6] {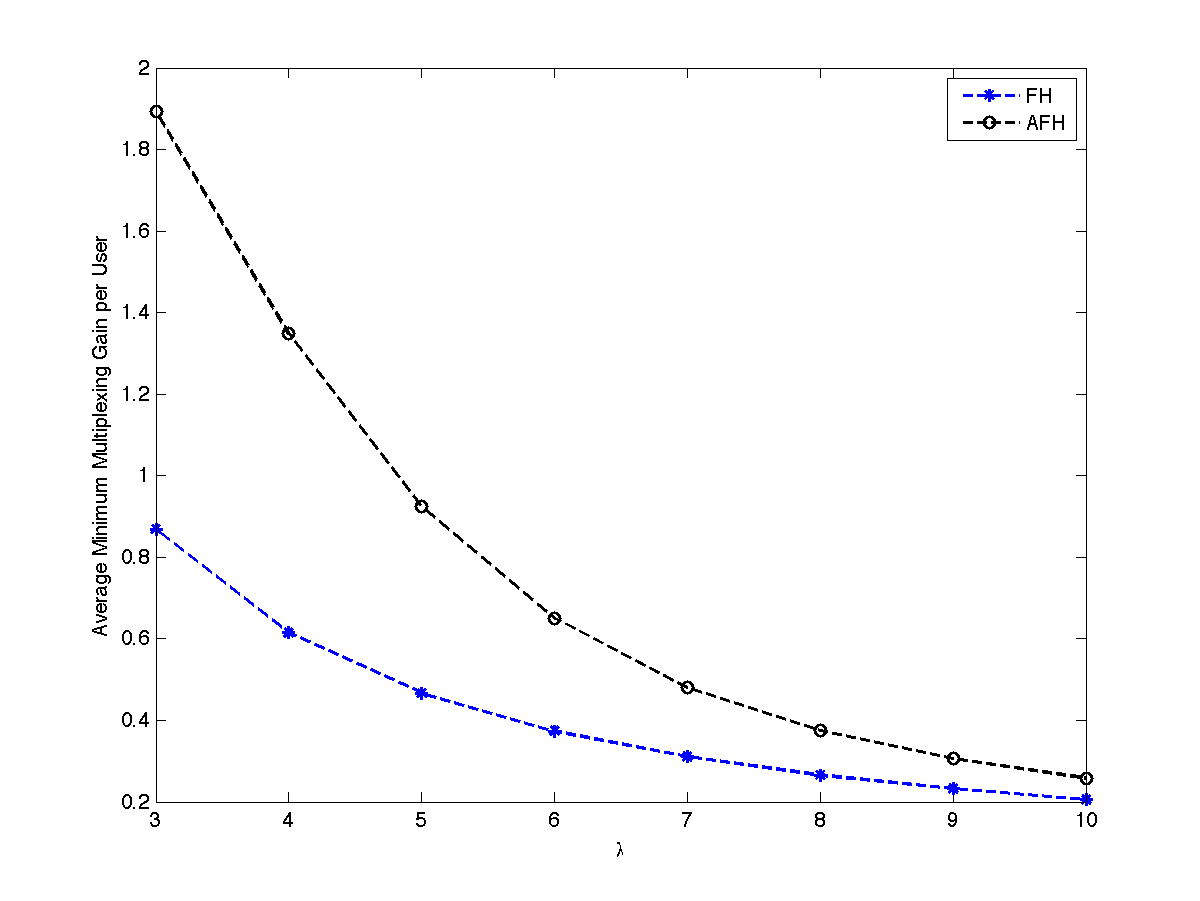}
  \caption{$\eta_{\mathrm{AFH}}^{(2)}$ versus $\eta_{\mathrm{FH}}^{(2)}$ for $u=10$.  }
  \label{fig7}
 \end{figure}

  \section{Conclusion}
We have addressed a decentralized  wireless communication network
with a fixed number $u$ of frequency sub-bands to be shared among
$N$ transmitter-receiver pairs. It is assumed that the number of
active users is a random variable with a given distribution.
Moreover, users are assumed to be unaware of each other's codebooks and hence, no
multiuser detection is possible.  We proposed a randomized Frequency
Hopping (FH) scheme in which each transmitter randomly hops over
subsets of the $u$ sub-bands from transmission to transmission.
Assuming all users transmit Gaussian signals, the distribution of
noise plus interference is mixed Gaussian, which makes
the calculation of the mutual information between the input and output
of each user intractable. We derived lower and upper bounds on this
mutual information and demonstrated that for large SNR values, the
two bounds coincide. This observation enabled us to compute the
sum multiplexing gain of the system and obtain the optimum
hopping strategy for maximizing this value. We compared the
performance of the  FH with that of the FD in terms of the following
performance measures: average sum multiplexing gain
$(\eta^{(1)})$, average minimum multiplexing gain per user $(\eta^{(2)})$,
minimum nonzero multiplexing gain per user $(\eta^{(3)})$ and
service capability ($\eta^{(4)}$). We showed that (depending on the
probability mass function of the number of active users) the FH
system can offer a significant improvement in terms of $\eta^{(1)}$
and $\eta^{(2)}$ (implying a more efficient usage of the spectrum).
It was also shown that
$\frac{1}{e}\leq\frac{\eta_{\mathrm{FH}}^{(3)}}{\eta_{\mathrm{FD}}^{(3)}}\leq
1$, i.e., the loss incurred in terms of $\eta^{(3)}$ is not more
than $\frac{1}{e}$. Moreover, computation of the so-called service
capability showed that in the FH system any number of users can
coexist fairly, while the maximum number of users in the FD system
is limited by the number of  sub-bands.

 \section*{Appendix A; Proof of Lemmas 2}
Let us consider a general $t\times 1$ vector mixed Gaussian
distribution $p_{\vec{\Theta}}(\vec{\theta})$ with different
covariance matrices $\{C_{l}\}_{l=1}^{L}$ and associated
probabilities $\{p_{l}\}_{l=1}^{L}$ given by
\begin{equation}
\label{gb}
p_{\vec{\Theta}}(\vec{\theta})=\sum_{l=1}^{L}p_{l}g_{t}(\vec{\theta},C_{l}),\end{equation}
where $g_{t}(\vec{\theta},C_{l})=\frac{1}{(2\pi)^{\frac{t}{2}}(\det C_{l})^{\frac{1}{2}}}\exp \left \{-\frac{1}{2}\vec{\theta}^{T}C_{l}^{-1}\vec{\theta} \right\}$.
 Hence,
\begin{equation} \label{pteta}
\int p_{\vec{\Theta}}(\vec{\theta})\log p_{\vec{\Theta}}(\vec{\theta})d\vec{\theta}=\sum_{l=1}^{L}J_{l}
\end{equation}
where $J_{l}\triangleq p_{l}\int g_{t}(\vec{\theta},C_{l})\log p_{\vec{\Theta}}(\vec{\theta})d\vec{\theta}$ for $1\leq l\leq L$.
To find a lower bound on $J_{l}$, we observe that
\begin{eqnarray}
\label{gbb}
J_{l} &=& p_{l}\int g_{t}(\vec{\theta},C_{l})\log\bigg(\sum_{m=1}^{L}p_{m}g_{t}(\vec{\theta},C_{m})\bigg)d\vec{\theta} \notag\\
&\geq& p_{l}\int g_{t}(\vec{\theta},C_{l})\log\big(p_{l}g_{t}(\vec{\theta},C_{l})\big)d\vec{\theta} \notag\\
&=& \left(p_{l}\log p_{l}\right)\int g_{t}(\vec{\theta},C_{l})d\vec{\theta}+p_{l}\int g_{t}(\vec{\theta},C_{l})\log g_{t}(\vec{\theta},C_{l})d\vec{\theta}\notag\\
&=& p_{l}\log p_{l}+p_{l}\int g_{t}(\vec{\theta},C_{l})\log g_{t}(\vec{\theta},C_{l})d\vec{\theta}\end{eqnarray}
Using this together with (\ref{pteta}) yields
\begin{eqnarray}
\mathrm{h}(\vec{\Theta}) &=& -\int p_{\vec{\Theta}}(\vec{\theta})\log p_{\vec{\Theta}}(\vec{\theta})d\vec{\theta}\notag\\&=&-\sum_{l=1}^{L}J_{l} \notag\\
&\leq&  -p_{l}\log p_{l}-p_{l}\int g_{t}(\vec{\theta},C_{l})\log g_{t}(\vec{\theta},C_{l})d\vec{\theta}\notag\\
&\stackrel{(a)}{=}&-\sum_{l=1}^{L}p_{l}\log
p_{l}+\frac{1}{2}\sum_{l=1}^{L}p_{l}\log\left((2\pi e)^{t}\det
C_{l}\right)\end{eqnarray}
where in $(a)$, we have used the fact that
the differential entropy of a $t\times 1$ Gaussian vector with
covariance matrix $C_{l}$ is $\frac{1}{2}\log\left((2\pi e)^{t}\det
C_{l}\right)$.

Let $t=1$ and $\vec{\Theta}=Z_{i,j}$. Therefore,
\begin{eqnarray}
 \mathrm{h}(Z_{i,j})&\leq& \frac{1}{2}\sum_{l=0}^{L_{i}}a_{i,l}\log(2\pi e\sigma_{i,l}^{2})-\sum_{l=0}^{L_{i}}a_{i,l}\log a_{i,l} \notag\\
&=& \frac{1}{2}\sum_{l=0}^{L_{i}}a_{i,l}\log(2\pi e\sigma_{i,0}^{2})+\frac{1}{2}\sum_{l=0}^{L_{i}}a_{i,l}\log\frac{\sigma_{i,l}^{2}}{\sigma_{i,0}^{2}}-\sum_{l=0}^{L_{i}}a_{i,l}\log a_{i,l}\notag\\
&=&\log(\sqrt{2\pi e}\sigma_{i,0})+\frac{1}{2}\sum_{l=1}^{L_{i}}a_{i,l}\log\frac{\sigma_{i,l}^{2}}{\sigma_{i,0}^{2}}-\sum_{l=0}^{L_{i}}a_{i,l}\log a_{i,l}\end{eqnarray}
However, for all $l\geq 1$, we have $\frac{\sigma_{i,l}^{2}}{\sigma_{i,0}^{2}}\leq \frac{\sigma_{i,L_{i}}^{2}}{\sigma_{i,0}^{2}}=1+c_{i,L_{i}}\gamma$. Thus,
\begin{eqnarray}
 \mathrm{h}(Z_{i,j})&\leq& \frac{1}{2}\sum_{l=1}^{L_{i}}a_{i,l}\log(1+c_{i,L_{i}}\gamma)+\log(\sqrt{2\pi e}\sigma_{i,0})-\sum_{l=0}^{L_{i}}a_{i,l}\log a_{i,l} \notag\\
&=& \frac{1}{2}(1-a_{i,0})\log(1+c_{i,L_{i}}\gamma)+\log(\sqrt{2\pi e}\sigma_{i,0})-\sum_{l=0}^{L_{i}}a_{i,l}\log a_{i,l}.
\end{eqnarray}
This concludes the proof of Lemma 2.
 \section*{Appendix B; Proof of Proposition 1}
We have
$\eta_{\mathrm{FD}}^{(1)}=\frac{\mathrm{E}\{N\}u}{2n_{\max}}$ and
$\eta_{\mathrm{FH}}^{(1)}=\frac{1}{2}\max_{v}\left\{v\mathrm{E}\left\{N\left(1-\frac{v}{u}\right)^{N-1}\right\}\right\}$.
Let us define $\Omega(v,N)\triangleq N\omega_{v}^{N-1}$ where
$\omega_{v}=1-\frac{v}{u}$. Thinking of $N$ as a real parameter for
the moment, we have $\frac{\partial^{2}}{\partial
N^{2}}\Omega(v,N)=\omega_{v}^{N-1}\left(N\left(\ln
\omega_{v}\right)^{2}+2\ln\omega_{v}\right)$. As $N\geq 1$, we have
$\frac{\partial^{2}}{\partial
N^{2}}\Omega(v,N)\geq\omega_{v}^{N-1}\left(\left(\ln
\omega_{v}\right)^{2}+2\ln\omega_{v}\right)$. But, $\left(\ln
\omega_{v}\right)^{2}+2\ln\omega_{v}\geq 0$ if and only if
$\omega_{v}\leq \frac{1}{e^{2}}$ or $\omega_{v}\geq 1$. Since
$\omega_{v}\leq 1$, we get $\omega_{v}\leq \frac{1}{e^{2}}$. This
implies that the function $\Omega(v,N)$ is a convex function of $N$
as far as $\omega_{v}\leq \frac{1}{e^{2}}$. Therefore, by Jensen's
inequality,
\begin{eqnarray}
\mathrm{E}\left\{N\left(1-\frac{v}{u}\right)^{N-1}\right\}=\mathrm{E}\{\Omega(v,N)\}\geq \Omega\left(v,\mathrm{E}\{N\}\right)=\mathrm{E}\{N\}\left(1-\frac{v}{u}\right)^{\mathrm{E}\{N\}-1}
\end{eqnarray}
which is valid as far as $v\geq \left(1-\frac{1}{e^{2}}\right)u$. Hence,
\begin{eqnarray}
\label{veg1}
\eta_{\mathrm{FH}}^{(1)}&=&\frac{1}{2}\max_{v}\left\{v\mathrm{E}\left\{N\left(1-\frac{v}{u}\right)^{N-1}\right\}\right\}\notag\\
&\geq&\frac{1}{2}\max_{v\in\left[\left(1-\frac{1}{e^{2}}\right)u,u\right]}\left\{v\mathrm{E}\left\{N\left(1-\frac{v}{u}\right)^{N-1}\right\}\right\}\notag\\
&\geq&\frac{1}{2}\mathrm{E}\{N\}
\max_{v\in\left[\left(1-\frac{1}{e^{2}}\right)u,u\right]}\left\{v\left(1-\frac{v}{u}\right)^{\mathrm{E}\{N\}-1}\right\}.
\end{eqnarray}
The function $v\left(1-\frac{v}{u}\right)^{\mathrm{E}\{N\}-1}$ is a
concave function in terms of $v$ that achieves its absolute maximum
at $\frac{u}{\mathrm{E}\{N\}}$. Therefore,
\begin{equation}
\label{veg2}
\max_{v\in\left[\left(1-\frac{1}{e^{2}}\right)u,u\right]}\left\{v\left(1-\frac{v}{u}\right)^{\mathrm{E}\{N\}-1}\right\}=\max\left\{1-\frac{1}{e^{2}},\frac{1}{\mathrm{E}\{N\}}\right\}\left(1-\max\left\{1-\frac{1}{e^{2}},\frac{1}{\mathrm{E}\{N\}}\right\}\right)^{\mathrm{E}\{N\}-1}u.
\end{equation}
Using (\ref{veg1}) and (\ref{veg2}),
\begin{equation}
\label{veg3}
\eta_{\mathrm{FH}}^{(1)}\geq \frac{1}{2}\max\left\{1-\frac{1}{e^{2}},\frac{1}{\mathrm{E}\{N\}}\right\}\left(1-\max\left\{1-\frac{1}{e^{2}},\frac{1}{\mathrm{E}\{N\}}\right\}\right)^{\mathrm{E}\{N\}-1}\mathrm{E}\{N\}u.\end{equation}
Hence, a sufficient condition for $\eta_{\mathrm{FH}}^{(1)}>\eta_{\mathrm{FD}}^{(1)}$ to hold is that
\begin{equation}
\label{veggg}
\max\left\{1-\frac{1}{e^{2}},\frac{1}{\mathrm{E}\{N\}}\right\}\left(1-\max\left\{1-\frac{1}{e^{2}},\frac{1}{\mathrm{E}\{N\}}\right\}\right)^{\mathrm{E}\{N\}-1}>\frac{1}{n_{\max}}.\end{equation}
If $\mathrm{E}\{N\}\geq \frac{e^{2}}{e^{2}-1}$, we have
$\max\left\{1-\frac{1}{e^{2}},\frac{1}{\mathrm{E}\{N\}}\right\}=1-\frac{1}{e^{2}}$.
Hence, (\ref{veggg}) reduces to the inequality
$\mathrm{E}\{N\}<\frac{1}{2}\ln\left((e^{2}-1)n_{\max}\right)$.
Therefore, if $\frac{e^{2}}{e^{2}-1}\leq\mathrm{E}\{N\}<
\frac{1}{2}\ln\left((e^{2}-1)n_{\max}\right)$, then (\ref{veggg}) is
satisfied. On the other hand, if $\mathrm{E}\{N\}\leq
\frac{e^{2}}{e^{2}-1}=1.1565$, we get
$\max\left\{1-\frac{1}{e^{2}},\frac{1}{\mathrm{E}\{N\}}\right\}=\frac{1}{\mathrm{E}\{N\}}$.
Thus, (\ref{veggg}) reduces to the inequality
$\frac{1}{\mathrm{E}\{N\}}\left(1-\frac{1}{\mathrm{E}\{N\}}\right)^{\mathrm{E}\{N\}-1}>\frac{1}{n_{\max}}$.
For each $n_{\max}\geq 2$, this yields an upper bound on
$\mathrm{E}\{N\}$. Since
$\frac{1}{\mathrm{E}\{N\}}\left(1-\frac{1}{\mathrm{E}\{N\}}\right)^{\mathrm{E}\{N\}-1}$
is a decreasing function of $\mathrm{E}\{N\}$, the smallest of these
upper bounds is obtained for $n_{\max}=2$ and is equal to $1.2938$.
This means that for $\mathrm{E}\{N\}\leq1.1565$, (\ref{veggg}) is automatically satisfied. Thus, (\ref{veggg}) is equivalent to
\begin{equation}
\mathrm{E}\{N\}< \frac{1}{2}\ln\left((e^{2}-1)n_{\max}\right). \label{dooly}\end{equation}
\section*{Appendix C; Proof of Proposition 2}
We have $\eta_{\mathrm{FD}}^{(2)}=\frac{u}{2n_{\max}}$ and
$\eta_{\mathrm{FH}}^{(2)}=\frac{1}{2}\max_{v}\left\{v\mathrm{E} \left \lbrace \left(1-\frac{v}{u}\right)^{N-1}\right \rbrace \right\}$.
The function $\left(1-\frac{v}{u}\right)^{N-1}$ is  convex  in terms
of $N$. Using Jenson's inequality,
\begin{equation}\label{oiuyt}\eta_{\mathrm{FH}}^{(2)}\geq
\frac{1}{2}\max_{v}\left\{v\left(1-\frac{v}{u}\right)^{\mathrm{E}\{N\}-1}\right\}=\frac{u}{2\mathrm{E}\{N\}}\left(1-\frac{1}{\mathrm{E}\{N\}}\right)^{\mathrm{E}\{N\}-1}.\end{equation}
Hence, a sufficient condition for
$\eta_{\mathrm{FH}}^{(2)}>\eta_{\mathrm{FD}}^{(2)}$ to hold is
\begin{equation}
\frac{1}{\mathrm{E}\{N\}}\left(1-\frac{1}{\mathrm{E}\{N\}}\right)^{\mathrm{E}\{N\}-1}>\frac{1}{n_{\max}}.\end{equation}

\section{Acknowledgments}\
The first author is indebted to M. A. Maddah-Ali and S. Oveis Gharan for their invaluable suggestions. 

    \bibliographystyle{IEEEbib}
      
\end{document}